\documentclass[final,,twoside]{IEEEtranTCOM}
\normalsize
\ifCLASSINFOpdf
\else
\fi
\usepackage{fancyhdr}
\usepackage{amsmath,amssymb}
\usepackage{amsfonts}
\usepackage{amsthm}
\usepackage{graphicx}
\usepackage{dsfont}
\usepackage{balance}
\usepackage{algpseudocode}
\usepackage{algorithm}

\usepackage{comment}

\usepackage{hyperref}    

\usepackage{color}
\usepackage{pgfplots}           
\pgfplotsset{compat=1.16}
\usepackage{subfigure}
\usepackage{caption}

\usepackage{enumerate}
\usepackage{graphics}
\usepackage{cite}
\usepackage{mathrsfs}
\usepackage{stfloats}
\usepackage{xfrac}

\usepackage{tikz}
\usepackage{mathdots}
\usepackage{yhmath}

\usepackage{array}
\usepackage{multirow}
\usepackage{textcomp}
\usepackage{gensymb}
\usepackage{tabularx}
\usepackage{booktabs}
\usepackage{graphicx}

\usepackage{listings}
\usepackage{xcolor}

\definecolor{codebg}{RGB}{248,248,248}
\definecolor{codeframe}{RGB}{180,180,180}

\usepackage{bigints}

\newtheorem{theorem}{Theorem}

\newtheorem{proposition}{Proposition}
\newtheorem{remark}{Remark}

\begin{document}

\title{\huge
Reliable ORIS-assisted FSO Communications via HARQ}

\author{Georgios D. Chondrogiannis,~\IEEEmembership{Graduate Student Member,~IEEE,}  Athanasios P. Chrysologou,~\IEEEmembership{Member,~IEEE}, Vasilis K. Papanikolaou,~\IEEEmembership{Member,~IEEE,} Alexandros-Apostolos A. Boulogeorgos,~\IEEEmembership{Senior Member,~IEEE,} \\ Nestor D. Chatzidiamantis,~\IEEEmembership{Member,~IEEE,} and Robert Schober,~\IEEEmembership{Fellow,~IEEE}

\thanks{G. D. Chondrogiannis and N. D. Chatzidiamantis are with the Department of Electrical and Computer Engineering, Aristotle University of Thessaloniki, 54124 Thessaloniki, Greece (e-mails:
\{gchondro, nestoras\}@auth.gr).}
\thanks{A. P. Chrysologou and A.-A. A. Boulogeorgos are  with the Department of Electrical and Computer Engineering, Democritus University of Thrace, 67100 Xanthi, Greece (e-mails: achrysol@ee.duth.gr, al.boulogeorgos@ieee.org).}

\thanks{V. K. Papanikolaou and R. Schober are with the Institute for Digital Communications (IDC), Friedrich-Alexander-University Erlangen-Nuremberg, 91058 Erlangen, Germany (e-mails: \{vasilis.papanikolaou, robert.schober\}@fau.de).}

}


\vspace{-0.3in}
\maketitle	

\begin{abstract}
This paper studies a free-space optical (FSO) link assisted by an optical reconfigurable intelligent surface (ORIS) and enhanced by a hybrid automatic repeat request (HARQ) scheme. The ORIS creates a virtual line-of-sight path around obstacles, while HARQ recovers frames corrupted by turbulence, pointing jitter, and geometric loss through retransmission and combining. We first derive a tractable statistical model for the end-to-end transmitter-ORIS-receiver (Tx–ORIS–Rx) reflected channel by jointly accounting for atmospheric turbulence, ORIS-induced pointing errors, and geometric attenuation. Building on these results, we obtain closed-form outage probability (OP) expressions for HARQ with Chase combining (HARQ-CC) and analytical outage upper bounds for HARQ with incremental redundancy (HARQ-IR), valid for an arbitrary maximum number of transmission rounds. We further conduct a high signal-to-noise ratio (SNR) analysis that provides a thorough characterization of the outage behavior and reveals the diversity order of both schemes. In addition, we characterize the delay behavior of the truncated HARQ process through the mean number of transmission rounds and the conditional mean number of rounds given successful decoding. Finally, numerical and Monte Carlo results validate the proposed analysis and show that HARQ substantially improves ORIS-assisted FSO reliability, with HARQ-IR achieving lower outage and delay than HARQ-CC, even for a small number of retransmission rounds.
\end{abstract}


\vspace{-0.2cm}
\section{Introduction} \label{S:Intro}


Free-space optical (FSO) communication has attracted significant interest as a high-capacity wireless technology due to its license-free operation, inherent directionality, security advantages, and ability to support high-data-rate transmission over moderate distances \cite{FSO_theory}. These features make FSO suitable for a wide range of terrestrial and space-oriented applications, including inter-building links, smart-city infrastructure, temporary high-capacity deployments, urban backhaul, and satellite-to-ground communications \cite{FSO_theory_2,Haas_Survey,FSO_app}. However, the performance and deployment flexibility of FSO systems are strongly constrained by the requirement of a clear line-of-sight (LoS) path between the optical transmitter (Tx) and receiver (Rx). In practical environments, this LoS path may be blocked or degraded by obstacles, atmospheric attenuation, turbulence-induced scintillation, and pointing errors, leading to severe fluctuations of the received optical power.

To overcome the LoS limitation, optical reconfigurable intelligent surfaces (ORISs) have recently emerged as a promising solution for FSO systems \cite{nature_1,nature_2}. By steering or reshaping the incident optical beam, an ORIS can establish a virtual LoS path between the Tx and Rx when the direct path is blocked. Unlike active relays, ORISs can be implemented using passive or nearly passive optical components, which makes them attractive for low-complexity deployment in obstructed environments \cite{schober_survey}. Depending on their implementation, ORISs can be broadly classified into mirror-array-based and metasurface-based designs. Mirror-array ORISs employ two-dimensional arrays of adjustable micro-mirrors to control the direction of the reflected beam \cite{Magaz,schober_survey}, whereas metasurface-based ORISs use subwavelength engineered elements to impose a programmable phase profile on the incident wavefront. This enables beam steering, wavefront shaping, and anomalous reflection according to the generalized reflection law \cite{roadmap,Optical_adaptive,elec_tun,beam_splitter}.

Recent works such as \cite{Alou_ORIS, Unified} model the metasurface as a discrete structure composed of discrete fading reflective elements, utilizing a statistical framework similar to the zero-boresight pointing error model. In particular, the authors of \cite{Alou_ORIS} derive compact approximate closed-form expressions for ORIS-assisted system performance by leveraging the central limit theorem, whereas the authors of \cite{Unified} present comprehensive analytical expressions applicable to both single and cascaded ORIS setups under various fading scenarios.
Conversely, studies such as \cite{Dobre_Haas, photonics_1, boulo_fso, Najafi_new, optics_express, ORIS-NOMA} model the reflecting surface as a continuous structure acting as an anomalous reflector. For instance, the authors of \cite{Dobre_Haas} apply a zero-boresight pointing error model while offering a detailed characterization of the receiver-side beam. Similarly, the authors of \cite{photonics_1} adopt the same model with a simplified treatment of beam displacement. The work in \cite{boulo_fso} assesses the performance of a cascaded multi-RIS FSO link assuming zero-boresight alignment at each stage. In \cite{Najafi_new}, the authors perform a statistical analysis of receiver-side pointing errors and model the end-to-end channel as a unified coefficient, also accounting for structural sway effects. Their proposed framework supports both 2D and 3D scenarios. Finally, the authors of\cite{optics_express} address outage-constrained transmission and power optimization utilizing the 2D pointing error model developed in \cite{Najafi_new}, while the authors of \cite{ORIS-NOMA} extend this framework to multiple-access scenarios.

Despite their ability to restore the LoS, ORIS-aided links remain vulnerable to random impairments such as atmospheric turbulence and pointing errors. These random fades translate to deep, rapidly varying signal-to-noise ratio (SNR) fluctuations that undermine the ultra-reliable, low-latency requirements of future 6G optical networks. In addition, the presence of the ORIS introduces an additional source of misalignment through surface or platform sway, which can intensify pointing losses. Hybrid automatic repeat request (HARQ) has been widely studied as an effective retransmission-based mechanism for improving the reliability of fading channels. In HARQ with chase combining (HARQ-CC), the same coded packet is retransmitted and combined with previous failed attempts, whereas in HARQ with incremental redundancy (HARQ-IR), each retransmission conveys additional parity information. In the context of FSO communications, HARQ schemes have been analyzed under turbulence-induced fading and pointing errors in \cite{HARQ_FSO} and \cite{improv}, while power allocation across HARQ rounds was investigated in \cite{Power-optimal_HARQ} using outage asymptotic analysis. HARQ has also been considered in several integrated optical/wireless architectures, e.g.,\cite{harq-fso-vehicular,HARQ_FSO_UAV}. However, to the best of the authors’ knowledge, HARQ-aided ORIS-assisted FSO links have not yet been analytically characterized. Moreover, the benefit of HARQ in ORIS-assisted FSO links is not obvious a priori, because while retransmissions improve reliability, they may also increase delay, and the dominant impairment may switch between turbulence and pointing errors.

Motivated by the aforementioned considerations, this work
investigates the performance of a HARQ-aided ORIS-assisted
FSO link. In contrast to ORIS models that represent
the reflected path as a cascade of independent fading hops or
as a composition of independently fading reflecting elements,
we adopt the continuous-surface 2D model of \cite{Najafi_new}, where the
Tx-ORIS-Rx path is treated as a single end-to-end (e2e) reflected
beam. In this model, the pointing loss is governed by the compounded
beam-footprint displacement in the photodetector (PD) plane, jointly
determined by Tx, ORIS, and Rx sway, as well as by the
incidence, reflection, and projection angles. This distinction
is particularly relevant when ORIS sway is non-negligible,
since its effect cannot be equivalently captured by multiplying
independent pointing-error coefficients over separate hops.
Moreover, unlike cascaded models, the adopted e2e model enables the analytical study of ORIS geometry, because it explicitly incorporates
the incidence, reflection, and projection angles throughout the
pointing-error model.
Building on this e2e channel representation, for
the first time, we derive tractable channel statistics and analyze
HARQ-CC and HARQ-IR in terms of outage probability (OP),
high-SNR diversity order, and delay-oriented performance
metrics. This allows us to quantify both the reliability gain of retransmissions and the latency cost of the truncated HARQ process.
Specifically, the main contributions of our work can be summarized as follows:
\begin{itemize}
  \item We develop a tractable statistical framework for the ORIS-assisted FSO e2e channel, jointly incorporating atmospheric turbulence, path loss, and compounded Tx/ORIS/Rx pointing errors induced by structural sway.
  \item We derive approximate closed-form expressions for the probability density function (PDF) and cumulative distribution function (CDF) of the e2e channel coefficient and validate their accuracy through Monte Carlo simulations and Kolmogorov-Smirnov (KS) goodness-of-fit tests.
  \item Using the proposed channel statistics, we obtain closed-form outage expressions for HARQ-CC and analytical OP upper bounds for HARQ-IR for an arbitrary retransmission budget.
  \item We conduct a high-SNR analysis revealing useful insights and showing that both HARQ-CC and HARQ-IR achieve the same diversity order, while HARQ-IR offers a coding-gain advantage.
  \item We analyze the delay behavior of the truncated HARQ-IR and HARQ-CC processes through both unconditional and conditional mean-round metrics, thereby separating the average retransmission cost of all packet attempts from the delay experienced by successfully decoded packets.
\end{itemize}

\emph{Notations}: The absolute value, exponential, and natural logarithm functions are, respectively, denoted by
$|{\cdot}|, \exp{(\cdot)}, \text{
and } \ln (x)$. $\mathbb{P}(A)$ denotes the probability that event $A$ occurs, and $\mathbb{E}[\cdot]$ denotes the expected value of a random variable. The operators $\max{(\cdot)}$, $\min{(\cdot)}$ and $\text{med}(\cdot)$ represent the maximum, the minimum and the median of a numerical set, respectively. Also, $
\left\{p_i\right\}_1^N=\left\{p_1, \cdots, p_N\right\}$ is the shorthand notation for a finite sequence of $N$ elements. Moreover, $K_\nu(x)$ and $\Gamma(x)$ represent the $\nu$-th order modified Bessel function of the second kind \cite[eq. (8.432/1)]{Gradshteyn2014} and the Gamma function \cite[eq. (8.310/1)]{Gradshteyn2014}, respectively. Finally, $G_{p, q}^{m, n}\left(\! x \Big{|} \begin{array}{c}\mathbf{a_p} \\ \mathbf{b_q}\end{array}\!\!\right)$ stands for the Meijer $G$-function \cite[eq.(9.301)]{Gradshteyn2014} and $H_{p, q}^{m, n}\! \left[z \,\bigg{|} \!\!\begin{array}{c} \left(a_1, A_1\right), \ldots, \left(a_p, A_p\right) \\ \left(b_1, B_1\right), \ldots, \left(b_q, B_q\right)\end{array}\!\!\right]$ denotes the (single-variate) Fox-$H$ function \cite[sec. 1.2]{H-function_Mathai}. 
\begin{figure}[t]
\centering
	\includegraphics[scale = 0.21]{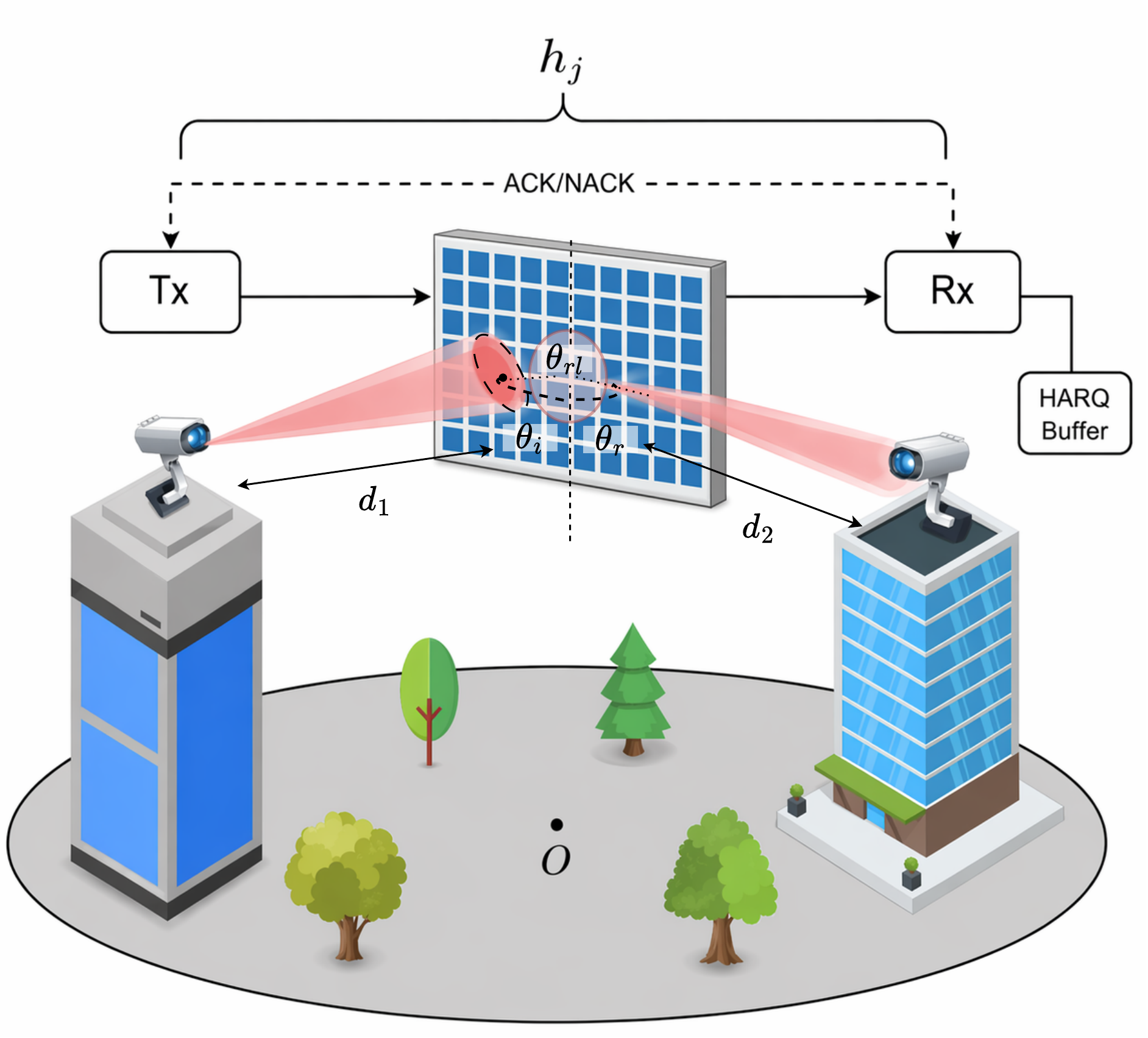}
 \caption{Illustration of ORIS-assisted FSO system with HARQ.}
	\label{Fig:draw}
\end{figure}
\section{System Model}
This section introduces a HARQ-aided ORIS-assisted FSO system model. We first provide the model for the intensity modulation direct detection (IM/DD) signal and the ORIS-assisted reflected optical path. Then, we present the statistical models for path-loss, Gamma–Gamma (GG) turbulence, and ORIS-induced pointing errors, which are later combined into a tractable e2e channel characterization.

\subsection{System and signal models}
As depicted in Fig. \ref{Fig:draw}, we consider an ORIS-assisted IM/DD FSO system which adopts HARQ to improve the reliability of the optical channel. In the present system,  the Tx encodes information by modulating the intensity of a laser beam, while the Rx employs a lens to focus the incoming optical signal onto a PD, which measures the signal's intensity to decode the transmitted data. The ORIS functions as a passive device for wavefront manipulation. By applying a configurable phase-shift distribution across its surface, it redirects the incoming optical beam toward a designated direction. The Tx, ORIS, and Rx are assumed to be installed on fixed structures, so that their nominal positions and the corresponding ORIS phase-shift profile remain unchanged over the considered transmission interval. Nevertheless, small random displacements caused by building sway are present around these nominal positions and are captured statistically through the adopted pointing-error model.

The signal at the optical receiver $y_m$ in the $m$-th HARQ round, for $m=1,2,\dots,M$, can be expressed as:
\begin{equation}
    y_m=  P x_mh_{m}+ n_m,
\end{equation}
where $x_m$ is the transmitted signal, $h_m$ is the e2e channel coefficient, $P$ is the optical transmit power of the transmitted signal, and $n_m$ is additive white gaussian noise (AWGN) with zero mean and variance $\sigma^2_{n}$.

\subsection{Channel and statistical models}
The e2e channel coefficient $h_m$ for the $m$-th HARQ round is considered to be the product of three factors \cite{Hralinovic,Najafi_new}: 
\begin{equation}
    h_m = h_l h_{am} h_{gm}.
    \label{ch_coefficient}
\end{equation}
Here, $h_l$ denotes the deterministic path-loss coefficient, $h_{gm}$ captures the pointing errors, and $h_{am}$ models the atmospheric turbulence at the Rx in the $m$-th round. Path-loss is assumed to remain constant over the considered HARQ transmission interval.
To simplify notation, we temporarily omit the HARQ-round index $m$ and write $h = h_l h_a h_g$. In what follows, we present the models adopted for $h_l$, $h_a$ and $h_g$.

\subsubsection{Atmospheric/Path-loss} The efficiency of practical reflecting surfaces can be compromised by the absorption or scattering of some of the incident beam's power. This phenomenon is dependent on operational parameters such as frequency and the bias voltage applied to the ORIS surface. To quantify path losses, the following equation is used
\begin{equation}
h_l=\zeta 10^{-\sigma d_{\mathrm{e2e}} / 10}.
\end{equation}
Here, $\zeta$ signifies the reflection efficiency, $\sigma$ denotes the attenuation coefficient, and $d_{\mathrm{e2e}}$ stands for the total e2e transmission distance of the optical beam. This distance is the sum of both the distance from the Tx to the ORIS, $d_1$, and from the ORIS to the receiver Rx, $d_2$, as illustrated in Fig. 1. For FSO links operating at a wavelength of $\lambda = 1550 \, \text{nm}$, typical reflection efficiency values for $\zeta$ lie between 0.7 and 1, as reported in \cite{Najafi_new}. The value of $\sigma$ depends on the current atmospheric conditions and can be estimated using Kim’s model described in \cite{Kim}.

\subsubsection{Fading}

Atmospheric turbulence causes the scintillation effect, resulting in signal fluctuations at the Rx. The GG model is widely acknowledged as an effective statistical model for describing various turbulence conditions, ranging from weak to strong. Its PDF is given as \cite{Dobre_Haas}
\begin{equation}
f_{h_a}(h_a) = \frac{2(\alpha \beta)^{\frac{\alpha+\beta}{2}}}{\Gamma(\alpha) \Gamma(\beta)} {h_a}^{\frac{\alpha+\beta}{2}-1} K_{\alpha-\beta}\left(2 \sqrt{\alpha \beta h_a}\right),
\label{f_ha}
\end{equation}
where $\alpha$ and $\beta$ are parameters related to large-scale and small-scale scattering effects, respectively. These parameters are given by:
\begin{align}
\alpha &= \left[\exp \left(\frac{0.49 \sigma_R^2}{\left(1+1.11 \sigma_R^{12 / 5}\right)^{7 / 6}}\right)-1\right]^{-1}, \nonumber \\
\beta &= \left[\exp \left(\frac{0.51 \sigma_R^2}{\left(1+0.69 \sigma_R^{12 / 5}\right)^{5 / 6}}\right)-1\right]^{-1},
\end{align}
where $\sigma_{R}^2 = 1.23 C_n^2 k^{7 / 6} d_{\mathrm{e2e}}^{11 / 6}$ represents the Rytov variance, $k = 2 \pi / \lambda$ is the wavenumber and $C_n^2$ is the refractive index structure parameter. For analytical tractability, we assume a block-fading model, where the channel remains approximately constant during each HARQ round, while consecutive retransmissions are treated as independent fading blocks. This assumption is justified when the packet duration is shorter than the coherence time of turbulence and pointing errors, and the feedback/retransmission process provides sufficient separation between rounds. In case of quasi-static fading, the combining/coding gain is still present but the time diversity across HARQ rounds is limited \cite{Power-optimal_HARQ}.

\subsubsection{Pointing Errors}
In ORIS-enabled links, the laser beam is directed toward the ORIS rather than directly to the Rx lens. Nevertheless, pointing errors can cause the center of the beam’s footprint on the ORIS to shift, resulting in alignment issues at the Rx plane. The extent of this misalignment is influenced by the phase-shifting characteristics or orientation of the ORIS, as discussed in \cite{schober_survey}.
When the positional fluctuations of the Tx, ORIS, and Rx are modeled as Gaussian, such as those caused by building sway, the PDF of the geometric loss component $h_g$ can be derived as \cite{Najafi_conf}
\begin{equation}
f_{h_g}\left(h_g\right)\!=\!\frac{\sqrt{\omega}}{ A_0 \sqrt{\pi}}\left[\ln \left(\frac{A_0}{h_g}\right)\right]^{-\frac{1}{2}}\!\left(\frac{h_g}{A_0}\right)^{\omega-1}, 
\end{equation}
where $0 \leq h_g \leq A_0$, with $A_0$ denoting the fraction of power collected at the center of the PD, which can be expressed as $A_0=\operatorname{erf}(\nu)$. Parameters $\omega$, $t$ and $\nu$ are given by $\omega = \frac{\tau \bar{w}^2\left(d_{e 2 e}\right)}{4 \sigma_{rx}^2}$, 
 $\tau=\frac{\sqrt{\pi} \operatorname{erf}(\nu)}{2 \nu \exp \left(-\nu^2\right) \cos ^2 \theta_{r p}}$, and $\nu=\frac{\sqrt{2} \cos \theta_{r p} a_l}{w\left(d_{e 2 e}\right)}$, respectively. Moreover, $\bar{w}\left(d_{\mathrm{e2e}}\right)$ is the beam width
at the Rx and $ \sigma_{rx}^2$ is the variance of the Gaussian misalignment vector on the PD plane, which is given by \cite{Najafi_conf}
\begin{equation}
    \sigma_{rx}^2=\frac{1}{\cos ^2 \theta_{r p}}\left(\frac{\cos ^2 \theta_r}{\cos ^2 \theta_i} \sigma_s^2+\frac{\sin ^2\left(\theta_i+\theta_r\right)}{\cos ^2 \theta_i} \sigma_r^2+\sigma_p^2\right).
\end{equation}
The variances $\sigma_s^2$, $\sigma_r^2$, and $\sigma_p^2$ characterize the building sways at the Tx, the ORIS and the Rx, respectively. Also, parameters $a_l$, $\theta_i$, $\theta_{r}$, and $\theta_{rp}$ represent, respectively, the radius of the PD, the angle of the incident beam at the ORIS, the angle of reflection between the laser beam and the ORIS, and finally the angle between the reflected beam and the PD plane. 

\vspace{-0.2cm}
\subsection{Combined channel model statistics }
\vspace{-0.1cm}
Toward enabling a comprehensive performance analysis of ORIS-assisted communication systems, the derivation of an analytical expression for the PDF of the e2e channel is crucial. In this context, the PDF of the e2e channel, characterized by the product of two independent random variables (RVs), \( h_a \) and \( h_g \), further scaled by the atmospheric attenuation factor \( h_l \), can be expressed as follows \cite{Hralinovic}:
\begin{equation}
    f_h\left(h\right) = \int_{\frac{h}{A_0 h_l}}^{\infty} \frac{1}{h_a h_l} f_{h_g}\left(\frac{h}{h_a h_l}\right)f_{h_a}\left(h_a\right)\mathrm{d}h_a.
    \label{f_h}
\end{equation}
In the following theorem, we present a closed-form approximate  expression for the PDF of the e2e channel.

\begin{theorem}
   The PDF of the e2e channel can be approximated as:
    \begin{align}
	f_{h}(h) &=\frac{\sqrt{\omega}\alpha\beta}{h_l A_0 \Gamma(\alpha) \Gamma(\beta)} \nonumber \\ 
    & \times \mathrm{G}_{1,3}^{3,0}\!\left(\!\alpha \beta \frac{h}{A_0 h_l} \! \mathrel{\Bigg|}  
    \begin{array}{c}
	   \!\!\! \omega - \frac{1}{2} \\ \!\!\! \alpha - 1,\, \beta - 1,\, \omega-1
	\end{array}\!\!\!\right),
    \label{eq:pdf}
\end{align}
where $\omega$ quantifies the severity of the pointing error effects in the ORIS-enabled link.
\end{theorem}
\begin{IEEEproof}
    The proof is provided in Appendix A.
\end{IEEEproof}
The approximation in \eqref{eq:pdf} is accurate over a wide range of values of $\omega$, covering the operating regimes of the considered ORIS-assisted FSO setup, with its accuracy monotonically improving as $\omega$ increases. This behavior is assessed in Section \ref{sec:results} through numerical validation under different physical scenarios. Appendix F further verifies that \eqref{eq:pdf} satisfies the PDF normalization condition within the accuracy of the adopted analytical approximation.

Building on the PDF expression in (\ref{eq:pdf}), the following theorem provides the CDF of the e2e link.
\begin{theorem}
    The CDF of the e2e channel can be approximately expressed as:
    \begin{align}
	F_{h}(h) =\frac{\sqrt{\omega}}{\Gamma(\alpha) \Gamma(\beta)} \mathrm{G}_{2,4}^{3,1}\!\left(\!\alpha \beta \frac{h}{A_0 h_l} \! \mathrel{\Bigg|}  
    \begin{array}{c}
	   \!\!\! 1, \,\omega + \frac{1}{2} \\ \!\!\! \alpha , \beta , \omega, 0
	\end{array}\!\!\!\right).
    \label{eq:cdf}
    \end{align}
\end{theorem}
\begin{IEEEproof}
     The CDF is obtained as $F_{h}(h) = \int_{0}^{h} f_{h}(x)\mathrm{d}x$, where the PDF is provided by \eqref{eq:pdf}. By initially employing \cite[(26)]{algorithm} alongside \cite[(9.31.5)]{Gradshteyn2014}, and subsequently performing suitable algebraic manipulations, \eqref{eq:cdf} can be obtained.
\end{IEEEproof}
\begin{proposition}
Based on the previous Theorems, we derive the PDF and the CDF of the SNR in round $m$ of the HARQ process as:
\begin{align}
    &f_{\gamma_m}(\gamma_m) =\frac{\sqrt{\omega}}{2\gamma_m\Gamma(\alpha) \Gamma(\beta)}  
\mathrm{G}_{1,3}^{3,0}\!\left(\!\frac{\alpha \beta }{A_{0} h_{l}}\sqrt{\frac{\gamma_{m}}{\bar{\gamma}}} \! \mathrel{\Bigg|}  
    \begin{array}{c}
	   \!\!\! \omega + \frac{1}{2} \\ \!\!\! \alpha, \, \beta, \, \omega
	\end{array}\!\!\!\right),
    \label{eq:pdf_snr} \\
    & F_{\gamma_m}(\gamma_m) =\frac{\sqrt{\omega}}{\Gamma(\alpha) \Gamma(\beta)} \mathrm{G}_{2,4}^{3,1}\!\left(\!\frac{\alpha \beta }{A_{0} h_{l}}\sqrt{\frac{\gamma_{m}}{\bar{\gamma}}} \! \mathrel{\Bigg|}  
    \begin{array}{c}
	   \!\!\! 1, \,\omega + \frac{1}{2} \\ \!\!\! \alpha, \, \beta, \, \omega, \, 0
	\end{array}\!\!\!\right).
   \label{eq:cdf_snr}
\end{align}
\end{proposition}
\begin{IEEEproof}
Performing the RV transformation $\gamma_m = \bar{\gamma} h^2$, where 
\begin{equation}
    \bar{\gamma} = \frac{P^2}{\sigma_n^2}
\end{equation}
denotes the transmit SNR, we have 
\begin{equation}
    f_{\gamma_{m}}(\gamma_m)=\frac{1}{2 \sqrt{\bar{\gamma} \gamma_m}}f_h\left(\sqrt{\frac{\gamma_m}{\bar{\gamma}}}\right) , \: \:  \gamma_m \geq 0,
    \label{var_trans}
\end{equation}
which, by substituting \eqref{eq:pdf} in \eqref{var_trans}, can be expressed as:
    \begin{align}
         &f_{\gamma_m}(\gamma_m) =\frac{\sqrt{\omega}\alpha\beta}{2\sqrt{\bar{\gamma}}h_l A_0 \Gamma(\alpha) \Gamma(\beta)}\gamma_{m}^{-\frac{1}{2}} \nonumber \\ 
    & \; \; \; \; \; \; \; \; \; \; \; \; \times \mathrm{G}_{1,3}^{3,0}\!\left(\!\frac{\alpha \beta }{A_{0} h_{l}}\sqrt{\frac{\gamma_{m}}{\bar{\gamma}}} \! \mathrel{\Bigg|}  
    \begin{array}{c}
	   \!\!\! \omega - \frac{1}{2} \\ \!\!\! \alpha - 1, \, \beta - 1,\ \omega-1
	\end{array}\!\!\!\right).
    \end{align}
By multiplying and dividing by $\gamma_m$ and then utilizing \cite[(9.31.5)]{Gradshteyn2014}, we obtain \eqref{eq:pdf_snr}. The CDF is obtained as $F_{\gamma_m}(\gamma_m) = \int_{0}^{\gamma_m} f_{\gamma_m}(x)\mathrm{d}x$. By initially employing \cite[(26)]{algorithm} alongside \cite[(9.31.5)]{Gradshteyn2014}, and subsequently performing suitable algebraic manipulations, \eqref{eq:cdf_snr} results. This concludes the proof.
\end{IEEEproof}

\section{Performance Analysis}
This section presents the proposed outage and delay analysis of the HARQ-aided ORIS-assisted FSO link. Using the approximate e2e channel statistics derived in Section II, we first obtain the OP of HARQ-CC and an analytical OP upper bound for HARQ-IR. We then study the high-SNR behavior of both schemes to extract their diversity order and conclude with delay-oriented metrics for the truncated HARQ process.
\subsection{HARQ-CC Outage Probability }
The accumulated mutual information (AMI) across $M$ HARQ-CC rounds in the ORIS point-to-point system, can be expressed as \cite{Diamantoulakis}
\begin{align}
    \mathcal{I}^{\text{cc}}_M = \frac{1}{2} \log_2 \left(1+ 
    c\sum_{m=1}^{M} \gamma_m\right),
\end{align}
where $c = e/2\pi $ accounts for the IM/DD signaling model.
The OP of the HARQ-CC-aided ORIS assisted communication system is calculated as:
\begin{align}
    P_{o,M}^{\mathrm{cc}} = \mathbb{P}\left(\mathcal{I}^{\mathrm{cc}}_M< \mathcal{R}_t\right) = \mathbb{P}\left(\sum_{m=1}^M \gamma_m \leq \frac{2^{2\mathcal{R}_t}-1}{c}\right),
\end{align}
where the quantity $(2^{2\mathcal{R}_t}-1)/c$ is the required post-combining SNR threshold for successful decoding at spectral efficiency $\mathcal{R}_t$, which also corresponds to the selected modulation-and-coding scheme operating point.
\begin{proposition}
    The OP of the HARQ-CC-aided ORIS assisted channel is given by \eqref{Out_CC} shown at the top of the next page.
\end{proposition}
\begin{IEEEproof}
    The proof is provided in Appendix B.
\end{IEEEproof}
\subsection{HARQ-IR Outage Probability}
The AMI across $M$ HARQ-IR rounds in the ORIS point-to-point system, can be expressed as
\begin{align}
    \mathcal{I}^{\text{ir}}_M = \frac{1}{2} \sum_{m=1}^{M}\log_2\left(1+ 
    c \gamma_m\right).
\end{align}
Thus, the OP of the HARQ-IR-aided ORIS assisted communication system can be calculated as:
\begin{align}
    P_{o,M}^{\text{ir}} = \mathbb{P}\left(\mathcal{I}^{\text{ir}}_M< \mathcal{R}_t\right) = \mathbb{P}\left(\sum_{m=1}^{M}\log_2\left(1+ 
    c \gamma_m\right)< 2\mathcal{R}_t\right ).
    \label{out_ir_P}
\end{align}
\begin{proposition}
    An upper bound on the OP of the HARQ-IR-aided ORIS assisted channel is given by \eqref{up_bound_ir} shown at the top of this page.
\end{proposition}
\begin{figure*} 
\begin{align}
        P_{o,M}^{\mathrm{cc}}\left(\mathcal{R}_t\right) = \Bigg( \frac{\sqrt{\omega}}{2\Gamma(\alpha) \Gamma(\beta)} \Bigg)^M
H\underbrace{_{0,1;\,2,3;\,\dots;\,2,3}^{{0,0;\,3,1;\,\dots;\,3,1}}}_{M+1}
\left[
    {\frac{\alpha \beta }{A_{0} h_l} \sqrt{\frac{2^{2\mathcal{R}_t}-1}{c\bar{\gamma}}}}^{(M)}
    \mathrel{\Bigg|}  
  \begin{matrix}
    - \\
    \big(0\, ; \underbrace{\frac{1}{2},\dots,\frac{1}{2}}_M \big)
  \end{matrix}
   \mathrel{\Bigg|} 
   \begin{matrix}
    \left\{\left(1,\frac{1}{2}\right), \left(\frac{1}{2}+ \omega,1\right)\right\}_1^M\\
    \left\{(\alpha,1),(\beta,1),(\omega,1)\right\}_1^M
  \end{matrix}
\right]
\label{Out_CC}
    \end{align}
\hrulefill
\begin{align}
     P_{o,M}^{\mathrm{ir,ub}} \left(\mathcal{R}_t\right) = \Bigg(\frac{\sqrt{\omega}}{\Gamma(\alpha) \Gamma(\beta)}\Bigg)^M \mathrm{G}_{M+1,3M+1}^{3M,1}\left( \left(\frac{\alpha \beta }{A_{0} h_l} \sqrt{\frac{2^{\frac{2\mathcal{R}_t}{M}}-1}{c\bar{\gamma}}}\right)^{M} \! \mathrel{\Bigg|} \begin{array}{c} \!\!\!
1, \left\{\omega+\frac{1}{2}\right\}_1^M  \\ \!\!\!
\left\{\alpha\right\}_1^M, \left\{\beta\right\}_1^M,  \left\{\omega\right\}_1^M,0
\end{array}\!\!\!\right)
     \label{up_bound_ir}
\end{align}
  \vspace{-0.4cm}  \hrulefill
\end{figure*}

\begin{IEEEproof}
    The proof is provided in Appendix C.
\end{IEEEproof}

\subsection{Asymptotic \& bound analysis}
In this section, we provide an asymptotic analysis regarding the Rx's outage behavior for high transmit SNR values. Towards this direction, we provide an analytical approximation for the OP of HARQ-IR in the high-SNR region and an order estimate for the OP of HARQ-CC in the high-SNR limit. The derived expressions are subsequently utilized to extract the diversity order of both schemes, which can be exploited to provide valuable insights for system design.

\begin{proposition}
An asymptotic order estimate for the OP of ORIS-enabled HARQ-CC systems in the high-SNR regime after $M$ transmissions is given by
\begin{equation}
    P_{o,M}^{\mathrm{cc}, \infty}\left(u_{\bar{\gamma},\mathcal{R}_t}\right) \simeq O \left({u_{\bar{\gamma},\mathcal{R}_t}}^{\!\!Mq^{*}} \left(\ln u_{\bar{\gamma},\mathcal{R}_t}\right)^{M(M-1)}\right),
    \label{order_estimate}
\end{equation}
where  $u_{\bar{\gamma},\mathcal{R}_t} =\frac{\alpha \beta }{A_{0} h_l} \sqrt{\frac{2^{2\mathcal{R}_t}-1}{c\bar{\gamma}}}$ and $q^{*} = \min(\alpha, \beta, \omega)$.
\end{proposition}
\begin{IEEEproof}
    The proof is provided in Appendix D.
\end{IEEEproof}
Accounting for the fact that polynomial decay of $\bar{\gamma} ^ {-Mq^{*}/2}$ dominates the logarithmic growth in the asymptotic limit, the multiplicative logarithmic term  modifies the prefactor but does not change the fundamental rate at which the
OP goes to zero as $\bar{\gamma} \rightarrow \infty$. Therefore, we arrive at the following remark.

\begin{remark}
The diversity order of the HARQ-CC-assisted ORIS-enabled channel, after $M$ transmission rounds, is equal to $\mathcal{D}_{\mathrm{cc}}^{\infty} = \frac{M}{2} \min\left(\alpha, \beta, \omega\right)$. The primary factor that influences the asymptotic decay of the OP at the Rx are the turbulence conditions when $\min\left(\alpha,\beta\right) < \omega$, whereas the overall pointing errors become the dominant factor affecting performance when $\omega< \min\left(\alpha,\beta\right)$.
\end{remark}

\begin{proposition}
The OP of ORIS-enabled HARQ-IR systems in the high-SNR regime after $M$ maximum transmissions can be upper-bounded by
\begin{align}
&P^{\mathrm{ir,ub}}_{o,M}(\mathcal{R}_t) \! \stackrel{ \bar{\gamma}\rightarrow \infty}{\simeq} \!
\frac{M^{M-1}}{q^{*}(M-1)!}
\left(
\frac{\sqrt{\omega}\,\Gamma(p^{*}-q^{*})\Gamma(r^{*}-q^{*})} 
{\Gamma(\alpha)\Gamma(\beta)\Gamma\!\left(\omega-q^{*}+\frac12\right)}
\right)^{\!M}  \nonumber \\
& \times (-\ln\Phi_{\bar{\gamma},\mathcal{R}_t})^{M-1}\Phi_{\bar{\gamma},\mathcal{R}_t}^{M q^{*}},
\label{high-snr-ir}
\end{align}
where we defined $\Phi_{\bar{\gamma},\mathcal{R}_t} =
\frac{\alpha\beta}{A_0 h_l}
\sqrt{\frac{2^{\frac{2R_t}{M}}-1}{c\,\bar{\gamma}}}$, $q^{*} 
 = \min(\alpha, \beta, \omega)$, $p^{*} = \max(\alpha, \beta, \omega)$ and $r^{*} = \text{med}(\alpha, \beta, \omega)$.
\end{proposition}

\begin{IEEEproof}
    The proof is provided in Appendix E.
\end{IEEEproof}
Similarly, since the polynomial decay term $\bar{\gamma}^{-Mq^{*}/2}$ dominates the logarithmic factor in the high-SNR limit, the diversity order is determined only by the exponent of $\bar{\gamma}$. This leads to the following remark. 
\begin{remark}
The diversity order of the HARQ-IR-assisted ORIS-enabled channel is also equal to $\mathcal{D}_{ir}^{\infty} = \frac{M}{2} \min\left(\alpha, \beta, \omega\right)$. The primary factor that influences the asymptotic decay of OP of the Rx are the turbulence conditions when $\min\left(\alpha,\beta\right) < \omega$, whereas the overall pointing errors become dominant when $\omega< \min\left(\alpha,\beta\right)$.
\end{remark}
From the above remark it follows that, when $\omega$ is small, the system becomes pointing-error limited, and improving mechanical stability, alignment accuracy, or ORIS placement may provide more benefits than simply increasing the transmit power or the retransmission budget $M$. 

\vspace{-0.2cm}
\subsection{Mean Delay} In addition to reliability, the latency cost of retransmissions
is a critical design consideration for ORIS-assisted links. While
increasing the maximum number of transmission rounds~$M$
reduces the OP, it also introduces additional
delay when the channel conditions are poor. Therefore, it is
important to quantify the delay behavior of the truncated
HARQ process and understand the trade-off between reliability
and latency.

Let $D \in \{1,  \dots,m, \dots, M\}$ denote the number of HARQ rounds used until the packet is either successfully decoded or the maximum retransmission budget is exhausted. Then, for both HARQ-CC and HARQ-IR, the mean number of HARQ rounds per packet attempt, which captures the average protocol delay regardless of the decoding outcome, is given by \cite{E_D}
\begin{equation}
\mathbb{E}[D] = 1+\sum_{m=1}^{M-1} P_{\mathrm{o},m}(\mathcal{R}_t).
\label{eq:mean_D}
\end{equation}
In $\mathbb{E}[D]$, a packet that is not decoded within the $M$-round limit is counted as occupying all $M$ HARQ rounds before being dropped. Hence, $\mathbb{E}[D]$ measures the average HARQ-round consumption per packet attempt, including both successful and unsuccessful transmissions.
Larger values of $M$ reduce the OP, but may increase the average number of HARQ rounds, since packets that would have been declared unsuccessful for smaller $M$ are now allowed to be decoded at later rounds. 

However, conditioning on successful delivery yields a refined metric that isolates the delay induced by packets that are eventually decoded correctly, namely
\begin{align}
\mathbb{E}[D \mid \mathrm{success}] &= \sum_{m=1}^{M} m\,\Pr(D=m\mid \mathrm{success}) \nonumber \\
&=\frac{\sum_{m=1}^{M} m \left(P_{\mathrm{o},m-1}(\mathcal{R}_t) - P_{\mathrm{o},m}(\mathcal{R}_t)\right)}
{1 - P_{\mathrm{o},M}(\mathcal{R}_t)}.
\label{eq:mean_D_success}
\end{align}
The numerator in \eqref{eq:mean_D_success} is a weighted sum in which
$P_{\mathrm{o},m-1}(\mathcal{R}_t)-P_{\mathrm{o},m}(\mathcal{R}_t)$
denotes the probability that decoding succeeds exactly at round $m$,
while the denominator normalizes over all successful attempts. Therefore, $\mathbb{E}[D\mid\mathrm{success}]$ excludes the dropped packets and averages only over packet attempts that are successfully decoded within the $M$-round limit. In this sense, comparing $\mathbb{E}[D]$ and $\mathbb{E}[D\mid\mathrm{success}]$ separates the useful retransmission efforts that lead to successful delivery from the rounds spent on packets that ultimately fail.

\section{Numerical Results\label{sec:results}} 
In this section, we analyze the performance of the considered HARQ-assisted ORIS-enabled communication system by presenting both theoretical \eqref{Out_CC} (th.), theoretical upper bound \eqref{up_bound_ir} (th. ub.), asymptotic upper bound \eqref{high-snr-ir} (asymp. ub.), and simulation (sim.) results. Unless indicated otherwise, the parameter values adopted for the extraction of the following figures can be found in Table \ref{table:par}. It is noted that the value $\theta_{rp}=\pi/3 \text{ rad}$ corresponds to a relatively oblique
projection geometry and is therefore a more challenging scenario. Since $\cos(\theta_{rp})=0.5$, this angle affects both the
collected-power coefficient $A_0$ and the pointing-error parameter
$\omega$, reducing the effective aperture projection and increasing the
sensitivity to misalignment. Also, the numerical evaluation of the multivariate Fox-$H$ function follows the methodology and the code presented in \cite{Integral_of_Fox}. Simulations were conducted in MATLAB$\copyright$ and also the Symbolic Math Toolbox was utilized to prevent overflow issues\footnote{Specifically, the \emph{vpa($\cdot$)} function can be utilized to leverage higher-precision arithmetic and thereby prevent floating-point overflows.}.
\begin{table}[h]
\caption{Simulation parameter values.}
\centering
 \resizebox{0.95\columnwidth}{!}{%
 \begin{tabular}{ccc}
 	\hline Parameters & Symbol & Value \\
 	\hline 
        Wavelength  & $\lambda$  & $1550 \text{ nm}$ \\
        Link range & $d_{\mathrm{e2e}}$ & $1\mathrm{~km}$\\
        Attenuation coefficient & $\sigma$  & $0.43\times 10^{-3} \,\text{dB/m}$  \\
        Reflection efficiency & $\zeta$  &  0.8\\
        Lens radius & $a_{l}$  & $2.5\text{ cm}$ \\
        Angle b/w PD and reflected beam  & $\theta_{rp}$  & $\pi/3$ rad \\
        Reflection angle at ORIS  & $\theta_{r}$  & $\pi/6$ rad \\
        Incident angle at ORIS  & $\theta_{i}$  & $\pi/6$ rad \\
        Building sways at Tx, ORIS, Rx, & $(\sigma_s, \sigma_r, \sigma_p)$ &  $(a_l,a_l,a_l)$ \\
        Index of refraction structure & $C_n^2$ & $5\times 10^{-14}$ $\text{ m}^{-2/3}$ \\
        Spectral efficiency & $\mathcal{R}_t$ & $2 \text{ bps/Hz}$\\
 	\hline
 \end{tabular}}
 \label{table:par}
\end{table}

\begin{figure}[h]
\centering
\vspace{-0.1cm}
\includegraphics[width=0.85\columnwidth]{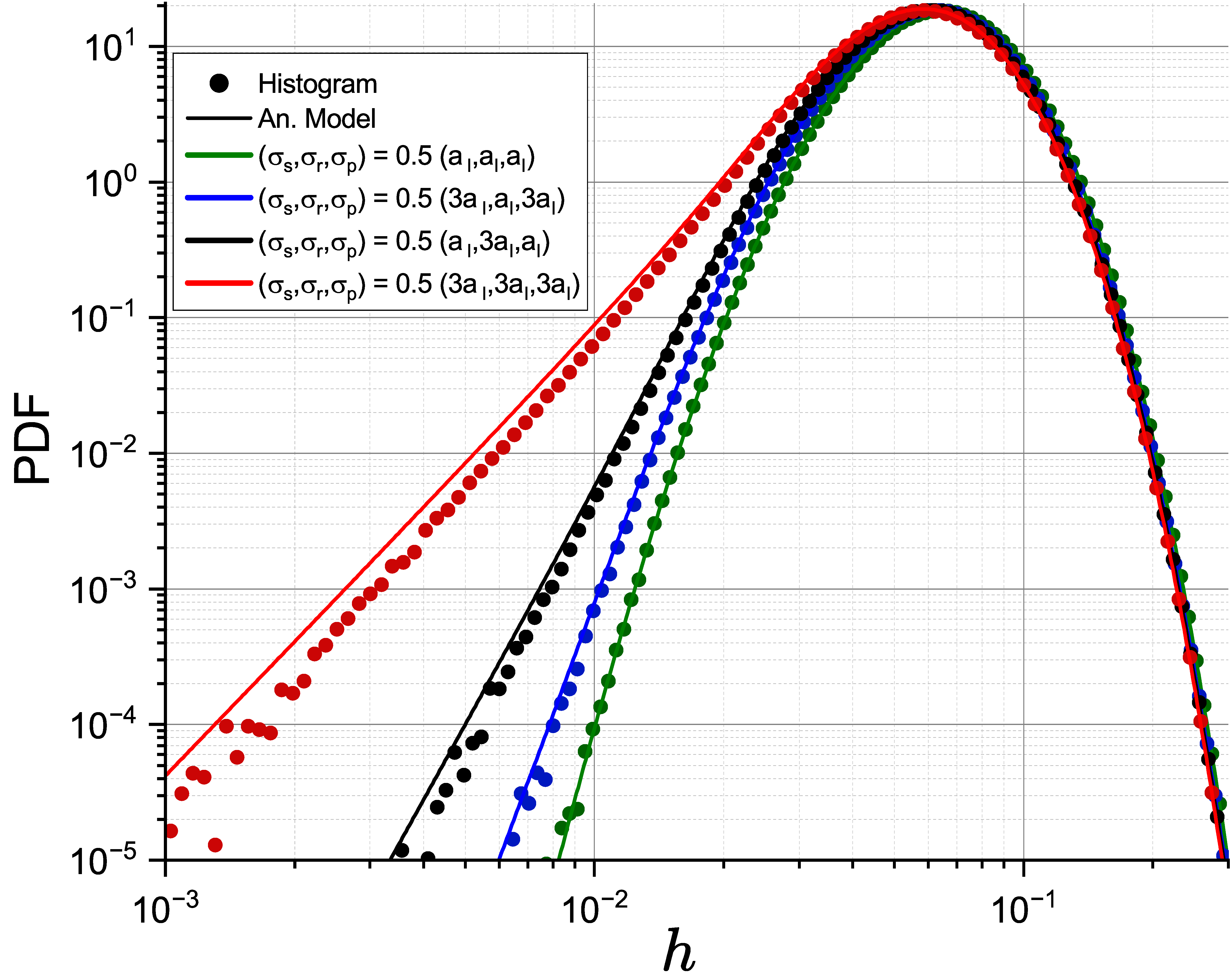}
\caption{PDF of the ORIS-enabled e2e channel.}
\vspace{-0.1cm}
\label{PDF_precision}
\end{figure}

\begin{figure}[h]
\centering
\includegraphics[width = 1\columnwidth]{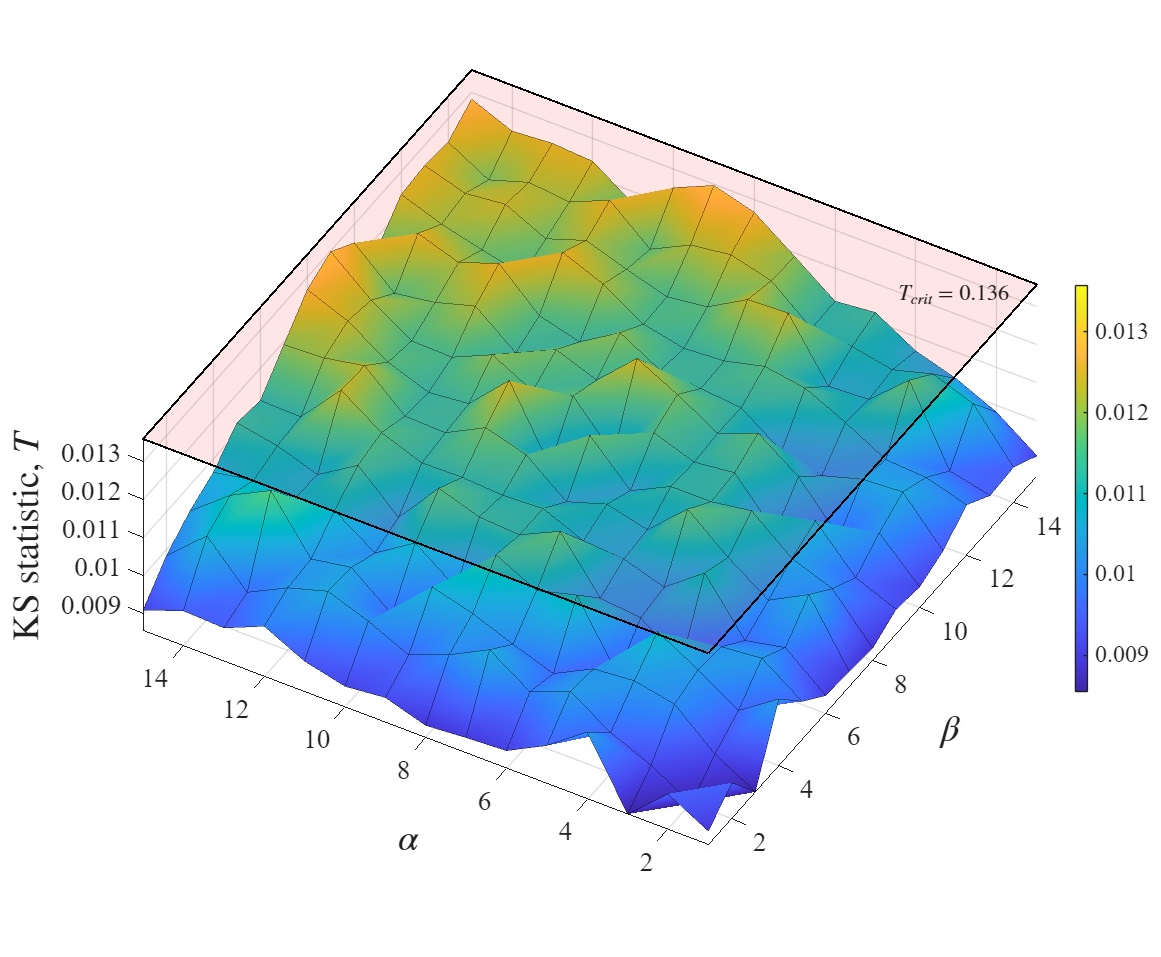}
\caption{ KS statistic, $T$, between the empirical CDF of the e2e channel coefficient, for different values of the GG turbulence parameters $\alpha$ and $\beta$. The horizontal plane denotes the reference critical value $T_{\text{crit}}$.}
\label{KS}
\end{figure}

In Fig. \ref{PDF_precision}, we illustrate the accuracy of the extracted statistical
model for the ORIS-enabled e2e channel, as characterized by \eqref{eq:pdf} and \eqref{eq:cdf}. For simulation, we plot the histogram of the channel coefficient $h$ for $10^8$ samples and 100 bins given by \eqref{ch_coefficient}. More specifically, Fig. \ref{PDF_precision} depicts the PDF of \eqref{eq:pdf} for four different building sway scenarios. Namely, (a) when Tx/Rx and the ORIS are subjected to the same modest building sways, (b) when the Tx and Rx sways are three times greater than the ORIS's, (c) when the ORIS's sway is three times greater than the Tx/Rx sways, and (d) when all three components are subjected to equal but more severe fluctuations. We observe that the PDF approximation exhibits close agreement with the Monte Carlo simulations in all considered cases, as evidenced by the close match between the analytical curves and the histograms of $h$. A slight deviation is observed in the severe-sway case, where stronger pointing errors correspond to a smaller pointing-error parameter $\omega$, while also considering deviations caused by the finite-sample binning. This is expected, since the approximation leading to \eqref{eq:pdf} becomes more accurate as $\omega$ increases. Nevertheless, the model still captures the dominant distributional behavior, including the shift of probability mass toward smaller values of $h$, which reflects a higher probability of severe pointing-induced attenuation and explains the outage degradation observed later.

To further assess the validity of the proposed approximate e2e channel model, we employ the KS goodness-of-fit test \cite{KS}. The KS test directly compares the empirical CDF obtained from Monte Carlo realizations with the CDF of the approximating model. Specifically, the KS statistic is defined as the maximum absolute distance between the empirical CDF, denoted by $\widehat{F}_h(\cdot)$, and the model CDF, denoted by $F_h(\cdot)$, i.e.,
\begin{equation}
T
\triangleq
\sup_{x\geq 0}
\left|
\widehat{F}_h(x)-F_h(x)
\right|.
\end{equation}
For each pair of turbulence parameters $(\alpha,\beta)$, Monte Carlo samples are generated for the exact e2e channel model, while the model CDF is obtained from \eqref{eq:cdf}. Fig.~\ref{KS} presents the mean KS statistic over the grid $\alpha,\beta\in\{1,\ldots,15\}$, which covers the majority of turbulence levels. The results are averaged over $60$ independent simulation runs, each with $N_{\mathrm{MC}}=10^4$ channel samples. The horizontal reference plane corresponds to the critical value \cite{KS}
\begin{equation}
T_{\text{crit}}
=
\frac{c_{a}}{\sqrt{N_{\mathrm{MC}}}},
\qquad
c_{a}
=
\sqrt{-\frac{1}{2}\ln\left(\frac{a_{\mathrm{sig}}}{2}\right)},
\end{equation}
with significance level $a_{\mathrm{sig}}=0.05$.
The KS surface remains under the reference plane over the entire turbulence range, confirming that the proposed approximation provides a satisfactory CDF fit. It is worth noting that this validation is performed for a small value of  $\omega\simeq 3.5$, which corresponds to a challenging misalignment regime. Therefore, the results show that the analytical approximation remains accurate even for low values of $\omega$ despite the initial stringent assumption of $\omega \gg 1$, while its accuracy is expected to further improve as $\omega$ increases. The KS statistic, $T$, slightly increases for larger values of $\alpha$ and $\beta$, where most channel realizations remain closer to the mean gain. As a result, small mismatches between the empirical and analytical CDFs can become more visible. Nevertheless, the discrepancy remains small, supporting the use of the proposed approximate channel statistics in the subsequent analysis.

\begin{figure}[h]
\centering
\vspace{-0.1cm}
\includegraphics[width=0.95\columnwidth]{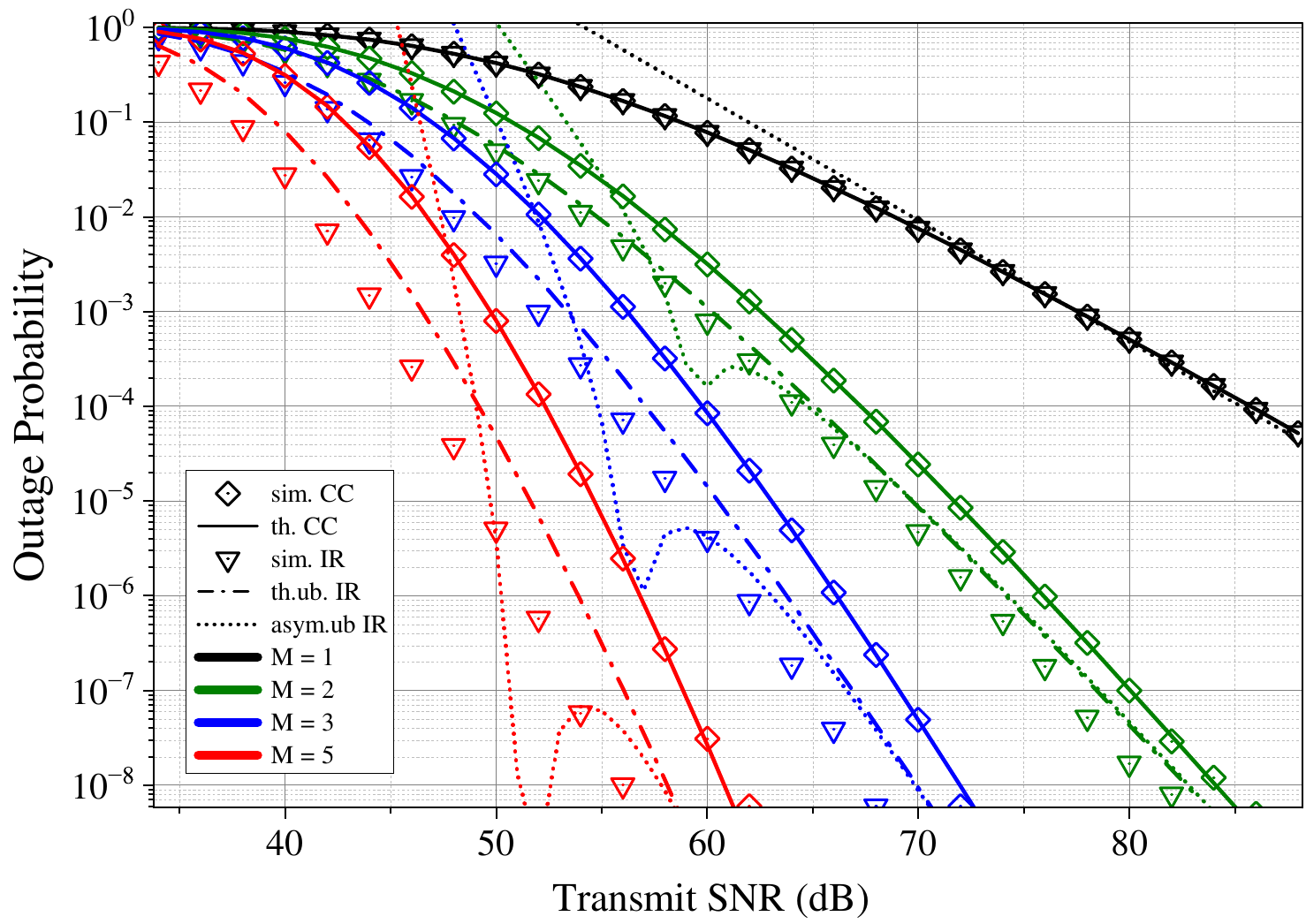}
\caption{OP of HARQ-CC and HARQ-IR versus transmit SNR, $\bar{\gamma}$, for different maximum numbers of HARQ transmissions $M$.}
\vspace{-0.1cm}
\label{OP_HARQ}
\end{figure}

In Fig. \ref{OP_HARQ}, the Rx’s OP versus the transmit SNR is illustrated for various maximum numbers of transmissions $M$
for both HARQ-CC and HARQ-IR. For the simulation-based outage evaluation, $10^9$ Monte Carlo samples were generated for each SNR point. In the case of HARQ-CC, the simulated curves closely follow the theoretical outage expression given by \eqref{Out_CC} across the full SNR range, validating the theoretical analysis. In the case of HARQ-IR, the derived upper bound follows the Monte Carlo trends and remains relatively close to the simulation results, especially for smaller values of $M$, e.g., $M=2$. For larger values of $M$, the bound becomes more conservative, while still capturing the correct outage behavior and high-SNR slope. It can also be observed that the asymptotic upper bound of \eqref{high-snr-ir} captures the correct high-SNR trend. However, it becomes accurate only at relatively high SNR. This happens since the dominant asymptotic term contains a logarithmic correction factor, as also discussed in Proposition 5. It is also evident that HARQ-IR outperforms HARQ-CC by achieving a lower OP, and this advantage becomes more pronounced as $M$ increases. For $M=1$, the two HARQ variants reduce to the same single-shot transmission case, which serves as consistency check. For $M>1$, the nearly parallel high-SNR slopes of HARQ-CC and HARQ-IR for the same $M$ are consistent with the analytical result in Remarks 1 and 2, where both schemes are shown to achieve the same diversity order. Hence, the advantage of HARQ-IR appears mainly as a coding-gain improvement. This behavior reflects the more efficient accumulation of mutual information in HARQ-IR across retransmission rounds. Finally, increasing 
$M$ steepens the outage decay for both schemes, in agreement with the derived diversity-order behavior.

\begin{figure}[h]
\centering
\vspace{-0.1cm}
\includegraphics[width=0.85\columnwidth]{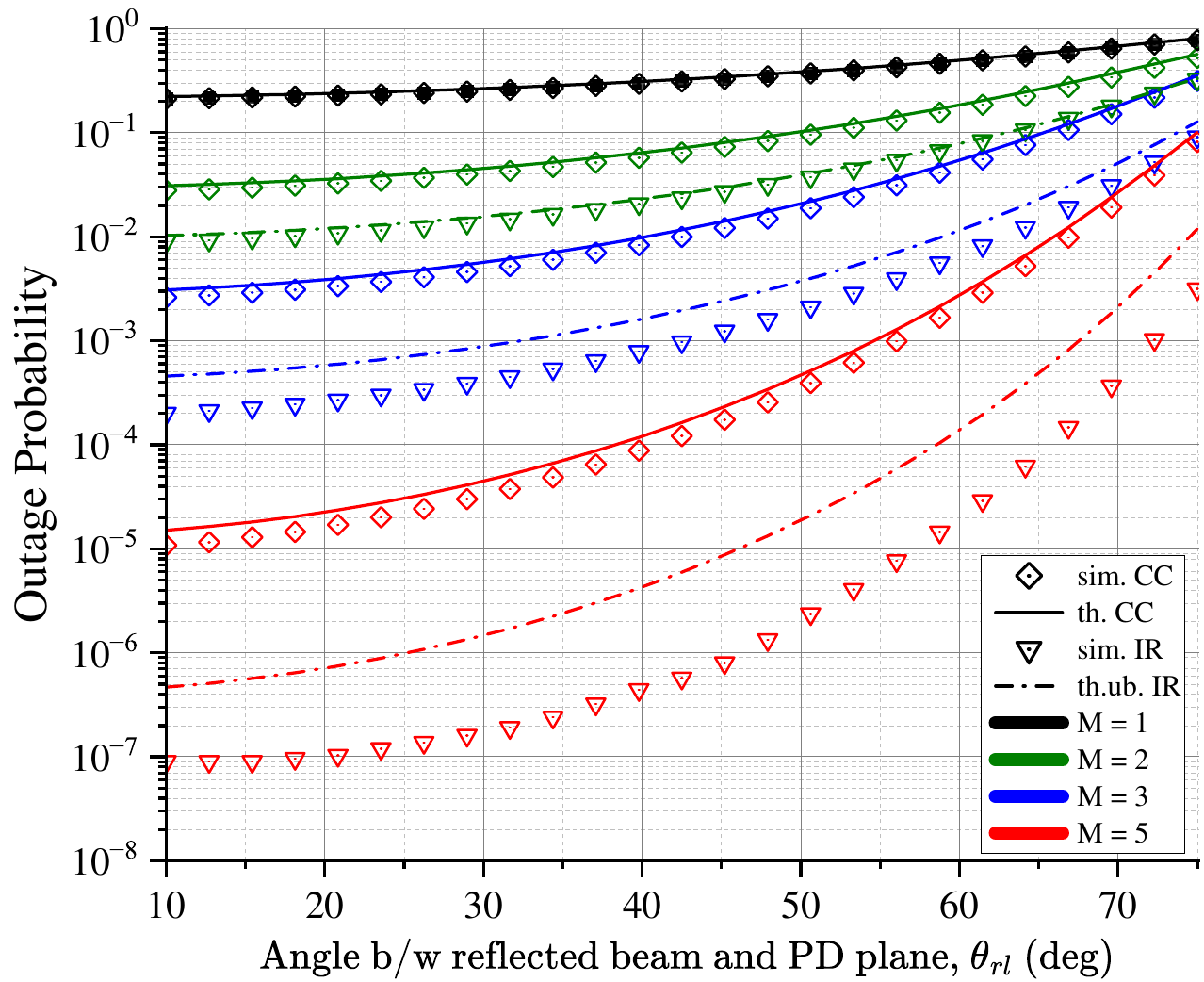}
\caption{OP of HARQ-CC and HARQ-IR versus $\theta_{rp}$, for different maximum HARQ transmissions $M$ with transmit SNR, $\bar{\gamma} = 50$ dB.}
\vspace{-0.1cm}
\label{theta}
\end{figure}

In Fig.~\ref{theta}, we investigate the effect of the ORIS reflection geometry by plotting the OP versus the angle between the reflected beam and the PD plane, $\theta_{rp}$. To emphasize the impact of severe pointing errors, we set the sway levels to $(\sigma_s,\sigma_r,\sigma_p)=(3a_l,3a_l,3a_l)$, while the remaining angular parameters are fixed as $\theta_i=\theta_r=\pi/6$. The angle $\theta_{rp}$ directly affects both the collected-power coefficient $A_0$ and the pointing-error parameter $\omega$. As $\theta_{rp}$ increases, the reflected beam is projected more obliquely onto the PD aperture, which reduces the effective collected power and increases the sensitivity to misalignment. As a result, the OP increases for both HARQ-CC and HARQ-IR.
The degradation becomes particularly pronounced for large values of $\theta_{rp}$, i.e., as the reflected beam approaches a grazing incidence with respect to the PD plane. In this regime, $\cos(\theta_{rp})$ becomes small, which strongly affects the projected beam footprint and increases the effective misalignment variance at the receiver plane. Consequently, even small structural displacements can produce severe pointing losses, causing the OP to rise rapidly as $\theta_{rp}$ approaches $90^\circ$. The model should therefore be interpreted over the physically meaningful range $\theta_{rp}<90^\circ$, since beyond this limit the reflected beam no longer impinges on the PD.
Nevertheless, HARQ-IR consistently outperforms HARQ-CC due to its more efficient accumulation of mutual information, while larger values of $M$ improve reliability through temporal diversity. The theoretical HARQ-CC curves closely follow the corresponding Monte Carlo simulations. For HARQ-IR, the derived upper bound captures the simulation trend over the considered angular range, although the gap between the bound and the Monte Carlo results becomes more pronounced for larger values of the maximum number of HARQ rounds, e.g. $M=5$. Therefore, the results validate the HARQ-CC analysis and show that the HARQ-IR bound provides an analytical upper-bound characterization of the geometry-dependent outage behavior.

\begin{figure}[h]
\centering
\vspace{-0.1cm}
\includegraphics[width=0.85\columnwidth]{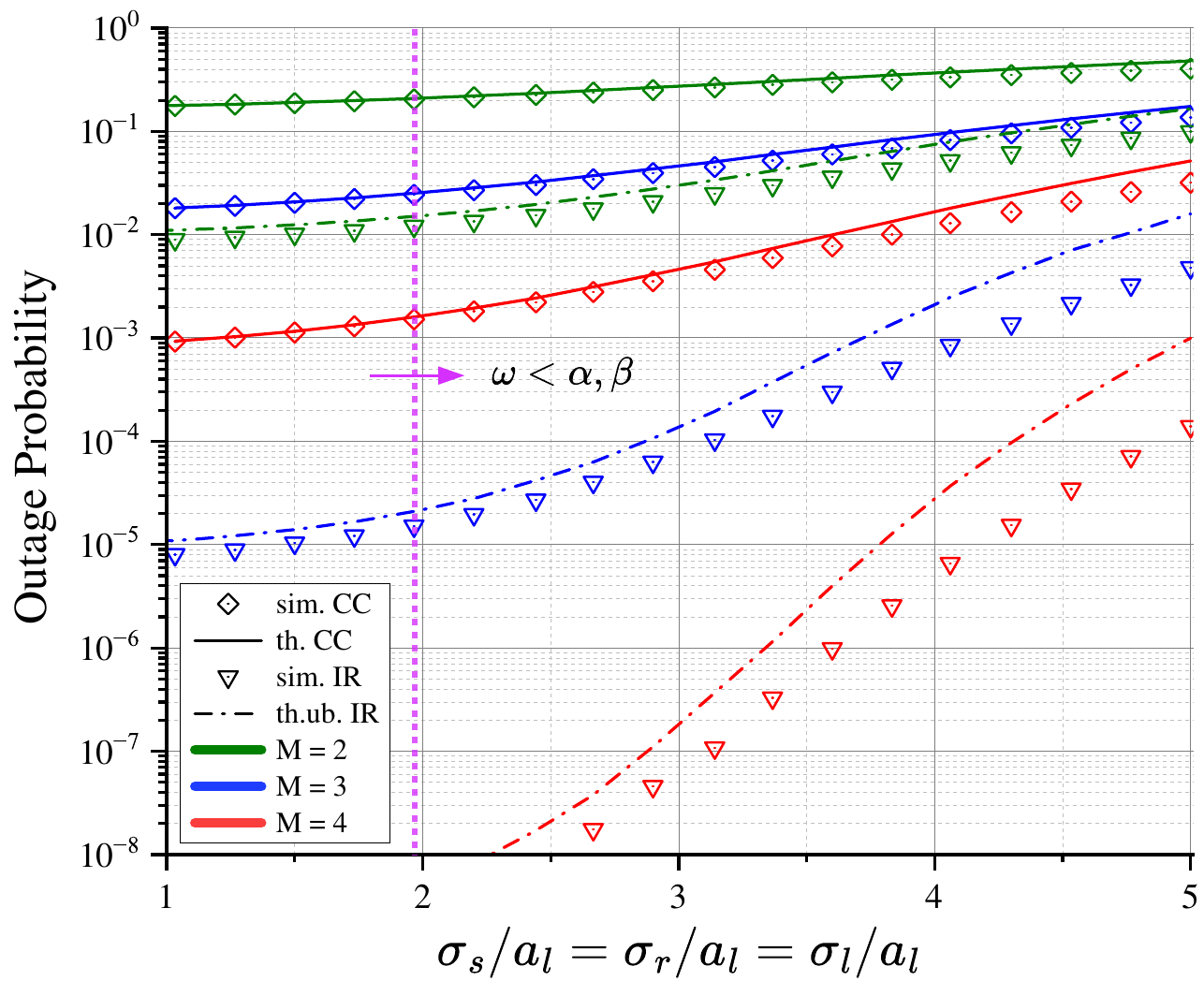}
\caption{OP of HARQ-CC and HARQ-IR versus common sway level for the Tx, Rx, and ORIS. Different numbers of maximum transmissions are plotted under fog conditions with transmit SNR, $\bar{\gamma} = 60$ dB.}
\vspace{-0.1cm}
\label{sigma}
\end{figure}

In Fig. \ref{sigma}, we investigate the impact of pointing errors on the OP by varying the normalized common sway level of the Tx, ORIS, and Rx. To emphasize the role of misalignment, the results are extracted under fog conditions with $\sigma=10\,\mathrm{dB/km}$ and $C_n^2=0.5\times10^{-14}\,\mathrm{m}^{-2/3}$. Under this configuration, turbulence is relatively weak, allowing the effect of structural sway and pointing errors to be more clearly observed.
As the common sway level increases, the pointing-error parameter $\omega$ decreases, which leads to stronger misalignment-induced attenuation and, consequently, higher OP. The vertical pink line, located slightly before $\sigma_r=2 a_l$, marks the transition point where pointing errors become the dominant impairment, i.e., where $\omega < \min\{\alpha, 
\beta\}$. As described in Remarks 1 and 2, beyond this point, the OP is mainly governed by pointing errors, and increases rapidly as the sway level further increases .
The figure also shows that in some regimes, HARQ-IR with a smaller retransmission budget can even outperform HARQ-CC with a larger one. For example, HARQ-IR with $M=3$ can achieve a lower OP than HARQ-CC with $M=4$. This illustrates that the coding-gain advantage of HARQ-IR can partly compensate for using fewer retransmission rounds.
This advantage becomes more pronounced under harsher propagation conditions. In these regimes, each retransmission is more valuable because the received information per round is more strongly impaired, and HARQ-IR exploits the available retransmission budget more efficiently. 
Finally, the agreement between the analytical curves and Monte Carlo simulations remains satisfactory over the plotted range, even when $\omega$ takes relatively small values. Some deviations appear only for very large jitter variances, corresponding to the smallest values of $\omega$, where the approximation becomes less accurate. Nevertheless, the analytical model still captures the dominant outage trend and the rapid degradation that occurs once the system enters the pointing-error-limited regime.

\begin{figure}[h]
\centering
\vspace{-0.1cm}
\includegraphics[width= 0.85\columnwidth]{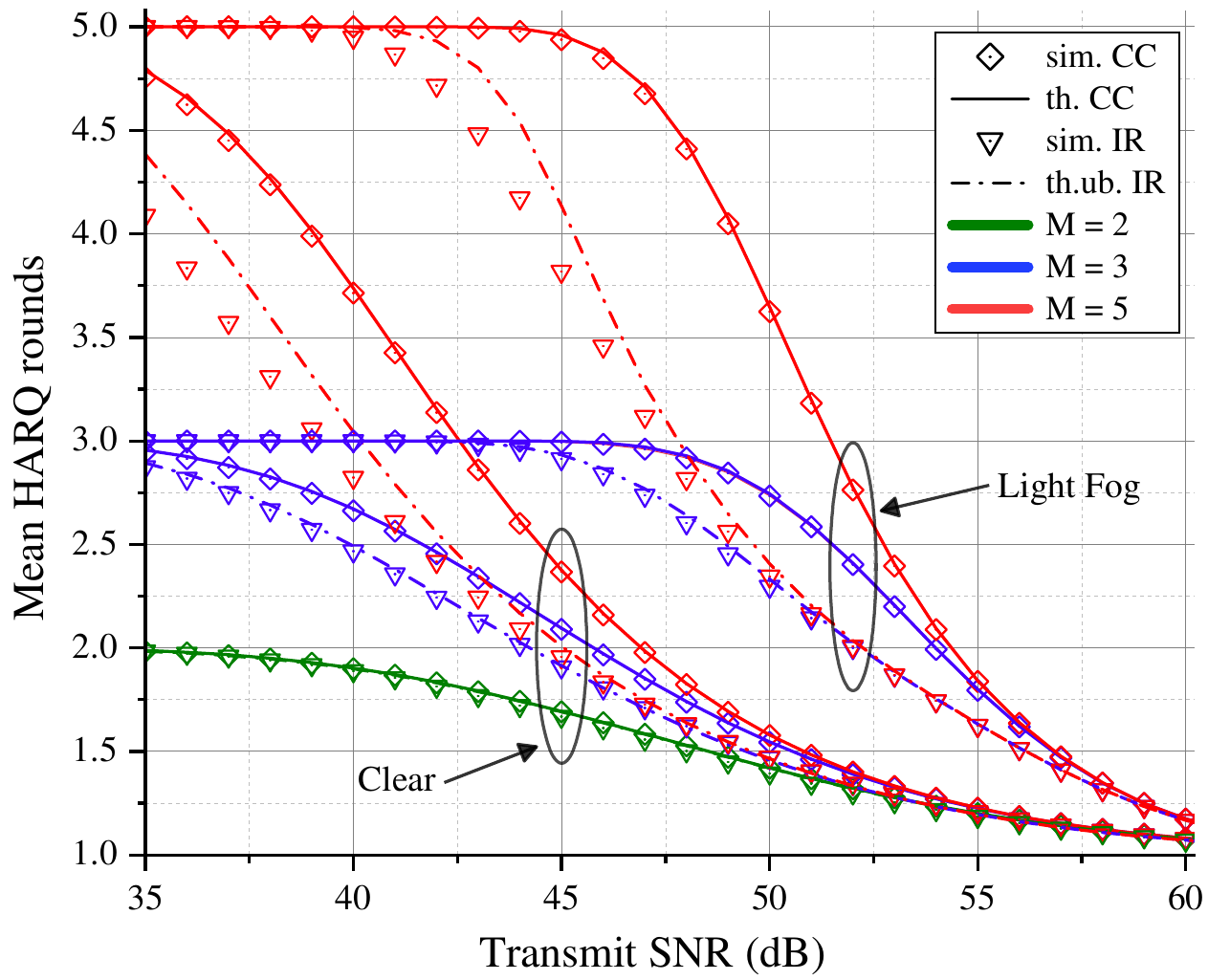}
\caption{Mean number of HARQ rounds versus transmit SNR $\bar{\gamma}$ for HARQ-CC and HARQ-IR, for different maximum number of transmissions $M$, under clear and light-fog conditions.}
\vspace{-0.1cm}
\label{E_D}
\end{figure}

In Fig. \ref{E_D}, we present both simulation and analytical results for the mean number of HARQ rounds versus the transmit SNR, until the transmitted message is correctly decoded or the maximum number of rounds $M$ is reached at the Rx. The results are extracted for different values of $M$ and under two atmospheric conditions, namely clear weather and light fog. The clear-weather case follows the baseline parameters in Table~I, whereas the light-fog case is modeled by using an attenuation coefficient of $\sigma=5\,\mathrm{dB/km}$ and a refractive-index structure parameter of $C_n^2=1\times10^{-14}\,\mathrm{m}^{-2/3}$.
In the low-SNR regime, successful decoding is unlikely before the retransmission budget is exhausted. Therefore, the mean number of rounds remains close to its corresponding maximum value $M$. As the transmit SNR increases, successful decoding occurs earlier and the mean delay decreases toward one round. This behavior is consistent with the outage trends in Fig. 4, where harsher propagation conditions make each additional retransmission round more valuable. For the same $M$ and given weather conditions, HARQ-IR enters the low-delay region earlier than HARQ-CC, showing that incremental redundancy not only improves reliability but also reduces the average number of required transmissions. This gain is small for $M=2$, where the retransmission budget is limited and HARQ-IR cannot fully exploit its incremental-redundancy advantage. By contrast, as $M$ increases, the performance gap becomes progressively more visible, reflecting the more efficient accumulation of mutual information in HARQ-IR across multiple retransmission rounds.

\begin{figure}[h]
\centering
\includegraphics[width= 0.86\columnwidth]{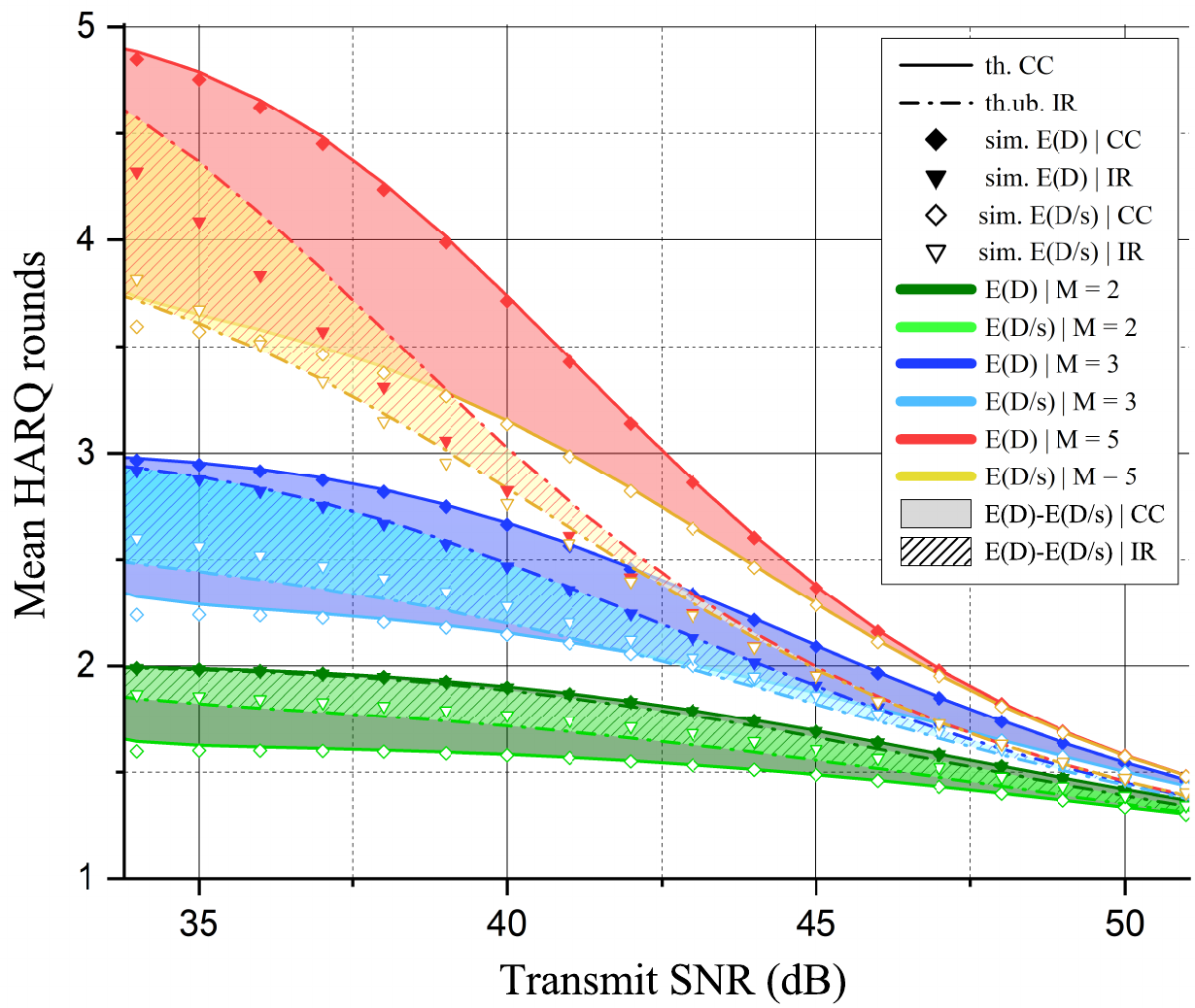}
\caption{Unconditional and conditional mean HARQ delay metrics
versus transmit SNR, for various $M$. The
shaded regions depict
$\mathbb{E}[D]-\mathbb{E}[D\mid\mathrm{success}]$.}
\label{E_D_suc}
\end{figure}

In Fig.~\ref{E_D_suc}, we jointly illustrate the unconditional mean number of HARQ rounds, $\mathbb{E}[D]$, and the conditional mean number of HARQ rounds given successful decoding, $\mathbb{E}[D\mid \mathrm{success}]$, as functions of the transmit SNR. The solid curves correspond to $\mathbb{E}[D]$, whereas the dashed curves correspond to $\mathbb{E}[D\mid \mathrm{success}]$.
The shaded regions quantify the difference between these two metrics. Specifically, the darker shaded regions enclosed by the corresponding HARQ-CC curves represent
$\Delta_D^{\mathrm{CC}}
=\mathbb{E}[D]_{\mathrm{CC}}-
\mathbb{E}[D\mid\mathrm{success}]_{\mathrm{CC}}$,
whereas the brighter shaded regions with line patterns enclosed by the HARQ-IR curves represent
$\Delta_D^{\mathrm{IR}}
=\mathbb{E}[D]_{\mathrm{IR}}-
\mathbb{E}[D\mid\mathrm{success}]_{\mathrm{IR}}$.
Therefore, these shaded regions provide a measure of the excess average round consumption associated with packet attempts that are ultimately dropped after the $M$-round limit is reached.
First, we observe that, for some retransmission budgets, especially $M=2$ and $M=3$, $\mathbb{E}[D\mid\mathrm{success}]_{\mathrm{IR}}$  can be  larger than that of HARQ-CC. This does not contradict the reliability advantage of HARQ-IR, but rather reflects the conditional nature of this metric. HARQ-IR can recover additional packets that would remain in outage under HARQ-CC, and some of these packets are decoded only in later retransmission rounds. Hence, the successful-packet ensemble of HARQ-IR may include more "challenging" channel realizations, which increase the conditional mean number of rounds. For larger budgets, such as $M=5$, HARQ-IR has more opportunities to exploit incremental redundancy over multiple rounds, so successful packets can be decoded more efficiently within the available retransmission window.
Moreover, at low SNR, both $\Delta_D^{\mathrm{CC}}$ and $\Delta_D^{\mathrm{IR}}$ are larger, indicating that a non-negligible fraction of the average round consumption is due to failed packet attempts. This effect becomes more intense as $M$ increases. Hence, although increasing $M$ improves reliability, it may also increase the retransmission effort spent on unsuccessful attempts when the SNR is too low.
Now, for the same retransmission budget, $\Delta_D^{\mathrm{CC}}$ is generally larger than $\Delta_D^{\mathrm{IR}}$, meaning that HARQ-CC spends a larger portion of its average round consumption on attempts that eventually fail.
As the transmit SNR increases, both shaded regions shrink since failed packet attempts become less frequent. 
Finally, we note that analytical values for $\mathbb{E}[D\mid\mathrm{success}]_{\mathrm{IR}}$ are obtained by substituting 
\eqref{up_bound_ir} into the conditional-delay expression. Thus, although \eqref{up_bound_ir} upper-bounds the OP, it is not guaranteed to upper-bound this nonlinear conditional delay ratio. Nevertheless, Fig.~\ref{E_D_suc} shows close agreement with the Monte Carlo results for all considered values of \(M\), indicating that the analytical bound remains sufficiently accurate for evaluating this metric.
\vspace{-0.1cm}
\section{Conclusions}\label{conc}
This paper investigated the reliability–delay performance of HARQ-aided ORIS-assisted FSO links. Although ORIS enables a virtual LoS path around blockages, the resulting link remains sensitive to atmospheric turbulence, path loss, and compounded e2e misalignment. Using the proposed statistical characterization of the Tx–ORIS–Rx channel, validated through Monte Carlo simulations and KS goodness-of-fit tests, analytical outage expressions were derived for HARQ-CC, while closed-form OP upper bounds were obtained for HARQ-IR. The high-SNR analysis and numerical results confirmed that both schemes achieve the same diversity order, which shows that the outage slope is governed by the strongest impairment among turbulence and pointing errors. Thus, in some cases, improving alignment may be more effective than only increasing transmit power or the retransmission budget. Numerical results showed that increasing $M$ improves outage performance through temporal diversity, while HARQ-IR outperforms HARQ-CC mainly due to its coding-gain advantage. The proposed delay analysis further showed that reliability improvements come with a retransmission cost: at low SNR, most packets consume the full HARQ budget, whereas at high SNR, the mean number of rounds approaches one. Finally, the gap between the unconditional and conditional mean number of rounds quantifies the retransmission effort spent on packets that ultimately fail, highlighting the need to jointly design the HARQ strategy, retransmission budget, ORIS geometry, and pointing stability.

 \vspace{-0.2cm} 
 \section*{Appendix A}
\label{app:thereom1}

 At first, combining \cite[(8.485)]{Gradshteyn2014} with \cite[(8.445)]{Gradshteyn2014} for $\alpha-\beta \notin \mathbb{Z}$, we can rewrite the Bessel function $K_{\alpha-\beta}\left(\cdot\right)$ in (\ref{f_ha}), in the following form: 
\begin{equation}
    K_{\alpha-\beta}\left(2 \sqrt{\alpha \beta h_a}\right)=\frac{1}{2} \mathrm{G}_{0,2}^{2,0}\left(\alpha \beta h_a \left\lvert \begin{array}{c}- \\ \frac{\alpha-\beta}{2} ,\frac{\beta-\alpha}{2}\!\!\!\end{array}\right.\right),
    \label{K_0}
\end{equation}
where we define $\chi = \frac{h}{A_0 h_l}$ for convenience. By utilizing (\ref{K_0}), as well as interchanging the order of summation and integration, which is justified by the existence of the PDF \footnote{Both $f_{h_a}(h_a)$ and $f_{h_g}(h_g)$ are well-defined and integrable over their respective domains \cite{Najafi_new}.} in (\ref{f_h}), we get
\begin{align}
    f_h\left(h\right) &= \mathcal{C} \chi^{\omega-1} \int_{\chi}^{\infty}  h_a^{\frac{a+b}{2}-\omega-1}\ln^{-\frac{1}{2}}\left(\frac{h_a}{\chi}\right) 
    \nonumber\\ & \times \mathrm{G}_{0,2}^{2,0}\left(\alpha \beta h_a \mathrel{\Big|} \begin{array}{c}- \\ \!\!\! \frac{\alpha-\beta}{2} ,\frac{\beta-\alpha}{2}\!\!\!\end{array}\right) \mathrm{d}h_a,
    \label{f_h_v2}
\end{align}
where 
\begin{equation}
    \mathcal{C} = \frac{\sqrt{\omega}\left(\alpha \beta\right)^{\frac{\alpha+\beta}{2}} }{\sqrt{\pi}A_0 h_l \Gamma(\alpha)\Gamma(\beta)}.
\end{equation}

Next, using the definition of the Meijer-$G$ function through Mellin–Barnes integral \cite[(9.301)]{Gradshteyn2014}, we can express the Meijer-$G$ function in (\ref{f_h_v2}) as
\begin{align}
    &\mathrm{G}_{0,2}^{2,0} \left(\alpha \beta h_a \mathrel{\Big|} \begin{array}{c}- \\ \!\!\! \frac{\alpha-\beta}{2} ,\frac{\beta-\alpha}{2}\!\!\!\end{array}\right) = \frac{1}{2\pi i} \int_{L} \Gamma\left(\frac{\alpha-\beta}{2}-s\right)
    \nonumber\\ & \times \Gamma\left(\frac{\beta-\alpha}{2}-s\right) \left(\alpha \beta h_a\right)^s \mathrm{d}s ,
    \label{MeijerG_repr}
\end{align}
where $L$ denotes the Mellin-Barnes integration contour in the complex $s$-plane.
By firstly substituting (\ref{MeijerG_repr}) in (\ref{f_h_v2}) and then interchanging the order of integration between the two integrals, we obtain
\begin{align}
    f_h\left(h\right) &=\mathcal{C} \chi^{\omega-1} 
    \nonumber\\ & \times \frac{1}{2 \pi i} \int_L \Gamma\left(\frac{\alpha-\beta}{2}-s\right) \Gamma\left(\frac{\beta-\alpha}{2}-s\right) (\alpha \beta )^s  \nonumber\\ & \times \underbrace{\left[\int_\chi^{\infty} h_a^{s+\frac{\alpha+\beta}{2}-\omega-1} \ln^{-\frac{1}{2}}\left(\frac{h_a}{\chi}\right) \mathrm{d}h_a\right]}_{I} \mathrm{d}s.
    \label{f_h_v3}
\end{align}
Consequently, if we perform the change of variables $y = \ln\left(h_a/\chi\right) $, the underlined integral in (\ref{f_h_v3}) can be evaluated as: 
\begin{align}
    I &= \chi^{s + \frac{\alpha+\beta}{2}-\omega }\int_0^{\infty} e^{y\left(s+\frac{\alpha+\beta}{2}-\omega\right)} y^{-\frac{1}{2}} \mathrm{d} y 
    \nonumber\\ &  = \frac{\sqrt{\pi}\, \chi^{s + \frac{\alpha+\beta}{2}-\omega}}{\left(\omega-s-\frac{\alpha+\beta}{2}\right)^{\frac{1}{2}}},
    \label{I}
\end{align}
where $\Gamma\left(\frac{1}{2}\right)= \sqrt{\pi}$.

Note that the convergence of \eqref{I} is ensured, regardless of the values of $\alpha$, $\beta$, and $\omega$, by the proper choice of the contour $L$ \cite{Luke}. Specifically, $L$ is taken as a vertical line $s=\delta+jt$, $t\in(-\infty,\infty)$, located to the left of the poles of the Gamma functions appearing in \eqref{f_h_v3} . Hence, $\delta$ is chosen to be smaller than the first pole of these sequences, i.e., $\delta<\omega-(\alpha+\beta)/2$, which guarantees the convergence of \eqref{I}.

Now, utilizing the Stirling's approximation for the Gamma function, which is given by
\begin{align}
   \Gamma(z) \sim \sqrt{2 \pi} e^{-z} z^{z-\frac{1}{2}} \text { for }|\arg (z)|<\pi,
\end{align}
and also assuming that $|1 /(2 z)| \ll 1$ and performing logarithmic transformations, we arrive at 
\begin{align}
    \frac{\Gamma(z)}{\Gamma\left(z+\frac{1}{2}\right)} = \frac{1}{\sqrt{z}}\left(1+\frac{1}{8 z}+O\left(\frac{1}{z^2}\right)\right),\,\,\, z \in \mathbb{C} \, ,
    \label{stirling_approx}
\end{align}
where $O(\cdot)$ denotes the Big-O asymptotic notation.
Thus, we can express \eqref{I} as
\begin{align}
    I \cong \frac{\sqrt{\pi}\, \chi^{s + \frac{\alpha+\beta}{2}-\omega} \, \Gamma\left(\omega-s-\frac{\alpha+\beta}{2}\right)}{\Gamma\left(\omega-s-\frac{\alpha+\beta}{2}+\frac{1}{2}\right)},
    \label{approxi}
\end{align}
under the typically satisfied condition
\(|\omega-s-\frac{\alpha+\beta}{2}| \gg 1\), particularly when $\omega \gg 1$. Since the
Mellin-Barnes contour is \(s=\delta+it\), the approximation is increasingly accurate when a valid contour is chosen such that
\(\omega-\delta-\frac{\alpha+\beta}{2}\gg 1\), while also satisfying
the convergence condition
\(\delta<\omega-\frac{\alpha+\beta}{2}\).
Now, substituting \eqref{approxi} in \eqref{f_h_v3}, we get 
\begin{align}
    &f_h\left(h\right) =\sqrt{\pi}\, \mathcal{C} \chi^{\frac{\alpha+\beta}{2}-1} \frac{1}{2 \pi i}\int_{L} (\alpha \beta \chi)^s
    \nonumber\\ & \times\!\frac{\Gamma\left(\frac{\alpha-\beta}{2}-s\right) \Gamma\left(\frac{\beta-\alpha}{2}-s\right)\Gamma\!\left(\omega-s-\frac{\alpha+\beta}{2}\right)}{\Gamma\!\left(\omega-s-\frac{\alpha+\beta}{2}+ \frac{1}{2}\right)} \mathrm{d}s.
    \label{poles}
\end{align}

The integral in (\ref{poles}) resembles the definition of Meijer-$G$ function through the Mellin-Barnes integral and as a result (\ref{poles}) becomes
\begin{align}
    f_h\left(h\right) \!=\!\sqrt{\pi} \, \mathcal{C} \chi^{\frac{\alpha+\beta}{2}-1}\, 
    \mathrm{G}_{1,3}^{3,0}\!\left(\!\alpha \beta \chi \! \mathrel{\Bigg|} \begin{array}{c}
	   \!\!\! \omega - \frac{\alpha+\beta}{2}+\frac{1}{2} \\ \!\!\! \frac{\alpha-\beta}{2}, \frac{\beta-\alpha}{2},\, \omega-\frac{\alpha+\beta}{2}
	\end{array}\!\!\!\right).
    \label{f_h_v4}
\end{align}
Finally, by substituting $\mathcal{C}$ and $\chi$ in (\ref{f_h_v4}) and then utilizing \cite[(9.31.5)]{Gradshteyn2014}, \eqref{eq:pdf} results.

\vspace{-0.1cm}
\section*{Appendix B}
First, we define the sum of SNR's up to the $M$-th round as 
\begin{equation}
    \mathcal{S}_M = \sum^{M}_{m=1} \gamma_m,
\end{equation}
then we proceed by calculating the moment generation function (MGF) of $\gamma_m$, i.e., $\mathcal{M}_{\gamma_m}\left(t\right)$, as 
\begin{align}
\mathcal{M}_{\gamma_m}\left(t\right) = \int_0^\infty e^{-tx} f_{\gamma_m}\left(x\right)dx.
\label{MGF}
\end{align}
Considering the Mellin-Barnes integral form of the Meijer-$G$ function in \eqref{eq:pdf_snr}, we can express \eqref{MGF} as 
\begin{align}
   &\mathcal{M}_{\gamma_m} \!\left(t\right) \!=\!  \frac{\sqrt{\omega}}{4\pi i\Gamma(\alpha) \Gamma(\beta)} \int_{L} \mathcal{\kappa}_{\bar{\gamma}}^s \frac{\Gamma\left(\alpha-s\right)\Gamma\left(\beta-s\right)\Gamma\left(\omega-s\right)}{\Gamma\left(\omega-s+\frac{1}{2}\right)} \nonumber \\ & \times \left(\int_0^\infty \!\! x^{\frac{s}{2}-1} e^{-tx} dx\right) \! \mathrm{d}s,
   \label{MGF_1}
\end{align}
where $\mathcal{\kappa}_{\bar{\gamma}} = \displaystyle \frac{\alpha \beta }{\sqrt{\bar{\gamma}}A_{0} h_l}$.
The inner integral can be calculated utilizing \cite[(3.381.4)]{Gradshteyn2014} as 
\begin{equation}
  \int_0^\infty x^{\frac{s}{2}-1} e^{-tx} \mathrm{d}x = \frac{1}{t^{s/2}} \Gamma\left(\frac{s}{2}\right).
  \label{int_s}
\end{equation}
Substituing \eqref{int_s} in \eqref{MGF_1} and applying the Mellin-Barnes definition of the Fox-$H$ function \cite[(1.1.1)]{H-transforms}, we arrive at:
\begin{align}
   \mathcal{M}_{\gamma_m}(t)=\frac{\sqrt{\omega}}{2\Gamma(\alpha)\Gamma(\beta)}{H}_{2,3}^{3,1}\!\left[\frac{\kappa_{\bar{\gamma}}}{\sqrt{t}}\,\Bigg|
\begin{array}{l}
\left(1,\frac{1}{2}\right),\left(\frac{1}{2}+\omega,1\right)\\
(\alpha,1),(\beta,1),(\omega,1)
\end{array}\!\!\!\right].
\label{Fox_1}
\end{align}
In order to extract the PDF of $\mathcal{S}_M$, we calculate the inverse Laplace transform of the product of MGFs in \eqref{Fox_1}, using the Mellin's inverse formula, as 
\begin{equation}
    f_{\mathcal{S}_M}(z)=\mathcal{L}^{-1} \left\{ \prod_{m=1}^M \mathcal{M}_{\gamma_m}(t)\right\}.
    \label{inverse_laplace}
\end{equation}
Thus, we substitute \eqref{MGF_1} in \eqref{inverse_laplace} and interchange the order of integration in order to express \eqref{inverse_laplace} as
\begin{align}
    & f_{\mathcal{S}_M}(z) = \Bigg( \frac{\sqrt{\omega}}{2\Gamma(\alpha) \Gamma(\beta)} \Bigg)^M \nonumber \\
    & \!\times \prod_{m=1}^M \frac{1}{2\pi i}\int_{L_m} \mathcal{\kappa}_{\bar{\gamma}}^{s_m} \frac{\Gamma\left(\alpha-s_m\right)\Gamma\left(\beta-s_m\right)\Gamma\left(\omega-s_m\right) \Gamma\left(\frac{s_m}{2}\right)}{\Gamma\left(\omega-s_m+\frac{1}{2}\right)}
     \nonumber \\
    & \!\times \Bigg( \frac{1}{2\pi i}\int_{\mathcal{T}} e^{zt}\, t^{-\frac{1}{2}\sum_{m=1}^M s_m}\mathrm{d}t\Bigg)\mathrm{d}s_m,
    \label{Multi_int}
\end{align}
where $L_m$ is an infinite contour in the complex $s_m$-plane, such that the integrand in \eqref{Multi_int} has no singularities. Also, $\mathcal{T}$ refers to  the contour path of the inverse Laplace transform integration,  so that $\mathcal{T}$ is in the region of convergence of $ t^{-\sum_{m=1}^M s_m/2}$. To calculate the inner integral, we utilize \cite[(8.315.1)]{Gradshteyn2014} as
\begin{equation}
    \int_{\mathcal{T}} t^{-\frac{1}{2}\sum_{m=1}^M s_m} \,e^{zt}\mathrm{d}t = \frac{2\pi i}{\Gamma\left(\frac{1}{2}\sum_{m=1}^M s_m\right)} \left(\frac{1}{z}\right)^{1-\frac{1}{2}\sum_{m=1}^M s_m}.
    \label{T}
\end{equation}
Substituting \eqref{T} in \eqref{Multi_int} and then applying the definition of $M$-variate Fox-$H$ function \cite[(A.1)]{H-function_Mathai}, we derive \eqref{Multi_Fox_PDF} shown at the top of the next page. 

In similar manner, the CDF of the $M$-th sum $\mathcal{S}_M$ is obtained as  $F_{\mathcal{S}_M}(z) = \int_{0}^{z} f_{\mathcal{S}_M}(x)\mathrm{d}x$ \footnote{Alternatively the inverse Laplace transform $F_{\mathcal{S}_M}(z)=\mathcal{L}^{-1} \left\{\prod_{m=1}^N \frac{M_{\gamma_m}(t)}{t}\right\}$ can be used to derive the CDF.}. Substituting \eqref{T} in \eqref{Multi_int} and interchanging the order of the integrals, we can express $F_{\mathcal{S}_M}(z)$ as 
\begin{align}
    & F_{\mathcal{S}_M}(z) = \Bigg( \frac{\sqrt{\omega}}{2\Gamma(\alpha) \Gamma(\beta)} \Bigg)^M \nonumber \\
    & \! \times \! \prod_{m=1}^M \frac{1}{2\pi i}\! \int_{L_m} \!\!\!\mathcal{\kappa}_{\bar{\gamma}}^{s_m} \frac{\Gamma\left(\alpha-s_m\right)\Gamma\left(\beta-s_m\right)\Gamma\left(\omega-s_m\right) \Gamma\left(\frac{s_m}{2}\right)}{\Gamma\left(\omega-s_m+\frac{1}{2}\right)\Gamma\left(\frac{1}{2}\sum_{m=1}^M s_m\right)}
     \nonumber \\
    & \! \times \Bigg(\int_{0}^z \left(\frac{1}{x}\right)^{1-\frac{1}{2}\sum_{m=1}^M s_m} \mathrm{d}x\Bigg)\mathrm{d}s_m.
    \label{Multi_int_cdf}
\end{align}
The inner integral in \eqref{Multi_int_cdf} is calculated as 
\begin{equation}
    \int_{0}^z \!\left(\frac{1}{x}\right)^{\!1-\frac{1}{2}\sum_{m=1}^M s_m} \mathrm{d}x \!=\! \frac{\Gamma\left(\frac{1}{2}\sum_{m=1}^M s_m\right)}{\Gamma\left(\frac{1}{2}\sum_{m=1}^M s_m+1\right)} z^{\frac{1}{2}\sum_{m=1}^M s_m},
    \label{integral_final}
\end{equation}
where we made use of $\frac{\Gamma\left(\frac{1}{2}\sum_{m=1}^M s_m\right)}{\Gamma\left(\frac{1}{2}\sum_{m=1}^M s_m+1\right)} = \frac{1}{\frac{1}{2}\sum_{m=1}^M s_m}$. Then, substituting \eqref{integral_final} in \eqref{Multi_int_cdf} and applying simple algebraic manipulations, while exploiting the definition of the $M$-variate Fox-$H$ function \cite[(A.1)]{H-function_Mathai}, we arrive at \eqref{Multi_Fox_CDF}. Finally, the OP for HARQ-CC case is given by $P_{o,M}^{\mathrm{cc}} = F_{\mathcal{S}_M}\left(\frac{2^{2\mathcal{R}_t}-1}{c}\right)$, or equivalently \eqref{Out_CC}. This concludes the proof.
\label{Harq-cc app}

\begin{figure*}[!t]
    \begin{equation}
f_{\mathcal{S}_M}(z) = \Bigg( \frac{\sqrt{\omega}}{2\Gamma(\alpha) \Gamma(\beta)} \Bigg)^M z^{-1}
H\underbrace{_{0,1;\,2,3;\,\dots;\,2,3}^{{0,0;\,3,1;\,\dots;\,3,1}}}_{M+1}
\left[
  \begin{matrix}
   {\mathcal{\kappa}_{\bar{\gamma}} \sqrt{z}}\, ^{(1)}\\
    \vdots\\
    {\mathcal{\kappa}_{\bar{\gamma}} \sqrt{z}}\, ^{(M)}
  \end{matrix}
  \mathrel{\Bigg|}  
  \begin{matrix}
    - \\
    \big(1\, ; \underbrace{\frac{1}{2},\dots,\frac{1}{2}}_M \big)
  \end{matrix}
   \mathrel{\Bigg|} 
   \begin{matrix}
    \left\{\left(1,\frac{1}{2}\right), \left(\frac{1}{2}+ \omega,1\right)\right\}_1^M\\
    \left\{(\alpha,1),(\beta,1),(\omega,1)\right\}_1^M
  \end{matrix}
\right]
\label{Multi_Fox_PDF}
\end{equation}
\hrulefill
\begin{equation}
F_{\mathcal{S}_M}(z) = \Bigg( \frac{\sqrt{\omega}}{2\Gamma(\alpha) \Gamma(\beta)} \Bigg)^M
H\underbrace{_{0,1;\,2,3;\,\dots;\,2,3}^{{0,0;\,3,1;\,\dots;\,3,1}}}_{M+1}
\left[
  \begin{matrix}
   {\mathcal{\kappa}_{\bar{\gamma}} \sqrt{z}}\, ^{(1)}\\
    \vdots\\
    {\mathcal{\kappa}_{\bar{\gamma}} \sqrt{z}}\, ^{(M)}
  \end{matrix}
    \mathrel{\Bigg|}  
  \begin{matrix}
    - \\
    \big(0\, ; \underbrace{\frac{1}{2},\dots,\frac{1}{2}}_M \big)
  \end{matrix}
   \mathrel{\Bigg|} 
   \begin{matrix}
    \left\{\left(1,\frac{1}{2}\right), \left(\frac{1}{2}+ \omega,1\right)\right\}_1^M\\
    \left\{(\alpha,1),(\beta,1),(\omega,1)\right\}_1^M
  \end{matrix}
\right]
\label{Multi_Fox_CDF}
\end{equation}
\hrulefill
\end{figure*}

\section*{Appendix C}
\label{Harq-ir app}
First, by utilizing the Minkowski inequality \cite{HARQ_FSO}, we have 
\begin{equation}
\left(\!\left(\prod_{m=1}^M \gamma_m\right)^{\frac{1}{M}}\!\!+1\!\right)^{\!\!M} \leq \prod_{m=1}^M\left(1+\gamma_m\right).
\end{equation}
Consequently, 
\begin{align}
    P_{o,M}^{\mathrm{ir}}\left(\mathcal{R}_t\right) & \le  \mathbb{P} \left\{M \log _2\left(1+\left(\prod_{m=1}^M c\gamma_m\right)^{\frac{1}{M}}\right) \leq 2\mathcal{R}_t\right\} \nonumber \\
    & = \mathbb{P} \left\{\prod_{m=1}^M \gamma_m \leq \left(\frac{2^{\frac{2\mathcal{R}_t}{M}}-1}{c}\right)^M\right\} = P_{o,M}^{\mathrm{ir,ub}}.
\end{align} 
The right-hand side serves as an upper bound on $P_{o,M}^{\mathrm{ir}}$ \cite{HARQ_FSO}.
Defining the product of SNR's up to the $M$-th HARQ round as
\begin{equation}
    \mathcal{Z}_M = \prod_{m=1}^M \gamma_m,
\end{equation}
we proceed with the calculation of the PDF of the product of $M=2$ channel SNR and then extend this to an arbitrary number $M$ of SNRs by induction. Specifically, the PDF of the product of two SNRs is calculated as
\begin{equation}
 f_{\mathcal{Z}_2}(z)=\int_{0}^{\infty} \frac{1}{x} f_{\gamma_1}(x) f_{\gamma_2}\left(\frac{z}{x}\right) \mathrm{d} x.
 \label{eq:14}
\end{equation}
After substituting \eqref{eq:pdf_snr} and applying the variable transformation $y = \sqrt{x}$ along with the Meijer-$G$ identity \cite[(07.34.21.0011.01)]{wolf}, then \eqref{int_z2} can be written as 
\begin{align}
    f_{\mathcal{Z}_2}(z) &= \frac{1}{2}\left(\frac{\sqrt{\omega}}{\Gamma(\alpha) \Gamma(\beta)}\right)^2 z^{-1} \nonumber \\ &\times\mathrm{G}_{2,6}^{6,0}\left( \left(\frac{\alpha \beta }{\sqrt{\bar{\gamma}}A_{0} h_l}\right)^2 \! \sqrt{z} \mathrel{\Bigg|} \begin{array}{c} \!\!\!
\omega+\frac{1}{2}, \omega+\frac{1}{2}  \\ \!\!\!
\alpha, \beta, \omega, \alpha, \beta, \omega
\end{array}\!\!\!\right).
    \label{res_z2}
\end{align}
By recurrently conducting this procedure, for $M = 3, 4, \dots,$ and taking into consideration the symmetries of the parameters appearing in Meijer-$G$ functions, we can prove that $f_{\mathcal{Z}_M}(z)$ is obtained as  
\begin{align}
    &f_{\mathcal{Z}_M}(z) = \frac{1}{2}\left(\frac{\sqrt{\omega}}{\Gamma(\alpha) \Gamma(\beta)}\right)^M z^{-1} \nonumber \\ &\times\mathrm{G}_{M,3M}^{3M,0}\left( \left(\frac{\alpha \beta }{\sqrt{\bar{\gamma}}A_{0} h_l}\right)^M \!\! \sqrt{z} \mathrel{\Bigg|} \begin{array}{c} \!\!\!
\left\{\omega+\frac{1}{2}\right\}_1^M  \\ \!\!\!
\left\{\alpha\right\}_1^M, \left\{\beta\right\}_1^M,  \left\{\omega\right\}_1^M
\end{array}\!\!\!\right).
    \label{res_zM}
\end{align}
 The CDF of the $M$-th product $\mathcal{Z}_M$ is obtained as $F_{\mathcal{Z}_M}(z) = \int_{0}^{z} f_{\mathcal{Z}_M}(x)\mathrm{d}x$. By first applying the substitution $t = \sqrt{x} $, then employing the identities provided in \cite[(26)]{algorithm} together with \cite[(9.31.5)]{Gradshteyn2014}, and finally conducting appropriate algebraic simplifications, the CDF can be formulated as shown in \eqref{res_CDF_zM}. Finally, the upper bound for the OP of HARQ-IR is given by $P_{o,M}^{\mathrm{ir,ub}} = F_{\mathcal{Z}_M}\left(\left(\frac{2^{\frac{2\mathcal{R}_t}{M}}-1}{c}\right)^{\!M}\right)$, or equivalently \eqref{up_bound_ir}. This concludes the proof.
 
\begin{figure*}[!t]
\vspace{-0.4cm}
\begin{align}
    f_{\mathcal{Z}_2}(z) = \frac{1}{z}\Bigg(\frac{\sqrt{\omega}}{2\Gamma(\alpha) \Gamma(\beta)}\Bigg)^2  \! \bigintssss_0^{\infty} \!\! x^{-1}  \mathrm{G}_{1,3}^{3,0}\left( \frac{\alpha \beta }{\sqrt{\bar{\gamma}}A_{0} h_l} \sqrt{x} \mathrel{\Bigg|} \begin{array}{c} \!\!\!
\omega+\frac{1}{2} \\ \!\!\!
\alpha, \beta, \omega
\end{array}\!\!\!\right) \mathrm{G}_{1,3}^{3,0}\left(\frac{\alpha \beta}{\sqrt{\bar{\gamma}}A_0 h_l } \sqrt{\frac{z}{x}}  \mathrel{\Bigg|}  \begin{array}{c}\!\!\!
\omega+\frac{1}{2} \\ \!\!\!
\alpha, \beta, \omega
\end{array}\!\!\!\right) dx
    \label{int_z2}
\end{align}
\hrulefill
\begin{align}
   F_{\mathcal{Z}_M}(z) = \Bigg(\frac{\sqrt{\omega}}{\Gamma(\alpha) \Gamma(\beta)}\Bigg)^M \mathrm{G}_{M+1,3M+1}^{3M,1}\left( \left(\frac{\alpha \beta }{\sqrt{\bar{\gamma}}A_{0} h_l}\right)^M \!\! \sqrt{z} \mathrel{\Bigg|} \begin{array}{c} \!\!\!
1, \left\{\omega+\frac{1}{2}\right\}_1^M  \\ \!\!\!
\left\{\alpha\right\}_1^M, \left\{\beta\right\}_1^M,  \left\{\omega\right\}_1^M,0
\end{array}\!\!\!\right)
    \label{res_CDF_zM} 
\end{align}
\hrulefill
\end{figure*}

\vspace{-0.1cm}
\section*{Appendix D}
\vspace{-0.1cm}
Regarding the code combining case, we consider the high-SNR regime where $\bar{\gamma} \rightarrow \infty$. In this regime, the arguments inside the $M$-variate Fox-$H$ function in \eqref{Out_CC} tend to zero. By defining $u = \frac{\alpha \beta }{A_{0} h_l} \sqrt{\frac{2^{2\mathcal{R}_t}-1}{c\bar{\gamma}}}$, we observe that as $u \rightarrow 0$, the asymptotic behavior of the Fox-$H$ function, based on \cite[(A.9)–(A.10)]{H-function_Mathai} in conjunction with \cite[Cor. 1.12.1]{H-transforms}, can be characterized as 
\begin{align}
&H \underbrace{_{0,1;\,2,3;\,\dots;\,2,3}^{{0,0;\,3,1;\,\dots;\,3,1}}}_{M+1}\Big[ u_1, u_2,\ldots,u_M \!\mathrel{\big|}\dots\mathrel{\big|}\dots\Big]  \nonumber 
\\ & = O\left(u_1^{q^{*}}\left(\ln u_1\right)^{M-1} u_2^{q^{*}}\left(\ln u_2\right)^{M-1} \! \cdots u_M^{q^{*}}\left(\ln u_M\right)^{M-1}\right),
\label{Big-O}
\end{align}
where $q^{*} = \min\left(\alpha,\beta,\omega\right)$ is the minimum value across the $b_q$ parameters/poles of \eqref{Big-O}.
Considering that $u =u_1 = u_2 = \ldots=u_M$, we can write \eqref{Big-O} as 
\begin{equation}
    H\Big[ u_1, u_2,\ldots,u_M\Big] \stackrel{u \rightarrow 0}{=}O \left(u^{Mq^{*}} \left(\ln u\right)^{M(M-1)}\right), 
\end{equation}
where $H\left[\cdot\right]$ follows the $M$-variate Fox-$H$ notation of \eqref{Out_CC}.
Correspondingly, based on \eqref{Out_CC}, we deduce that the OP for the HARQ-CC case, in the asymptotic limit of high SNR, 
OP scales as $u^{Mq^{*}} \left(\ln u\right)^{M(M-1)}$, as stated in \eqref{order_estimate}. This concludes the proof.

\vspace{-0.1cm}
\section*{Appendix E}
\vspace{-0.1cm}
In order to derive the high-SNR asymptotic behavior of the closed-form HARQ-IR upper bound in \eqref{up_bound_ir}, we first define
\begin{equation}
\Phi \triangleq
\frac{\alpha\beta}{A_0 h_l}
\sqrt{\frac{2^{\frac{2R_t}{M}}-1}{c\,\bar{\gamma}}}.
\end{equation}
As $\bar{\gamma}\to\infty$, we have $\Phi\to 0^+$  and thus the dominant term is determined by the minimum among the lower parameter set of Meijer-$G$ function in \eqref{up_bound_ir}, namely
$q^{*}=\min(\alpha,\beta,\omega)$.
So, the associated dominant pole at $s=-q^{*}$ is repeated $M$ times and therefore has multiplicity $M$. Using the repeated-pole expansion in \cite[(1.4.6)]{H-transforms}, a pole of multiplicity $M$ produces logarithmic terms up to order $M-1$. Hence, the dominant contribution is the term proportional to $\left(\ln \Phi^M\right)^{M-1}$. Consequently, as $\Phi\rightarrow 0$, the asymptotic Meijer-$G$ expansion is derived as
\begin{align}
&\mathrm{G}_{M+1,3M+1}^{3M,1}\left( \Phi^{M} \! \mathrel{|}\dots\right) \cong
\frac{(-1)^{M-1}}{q^{*}(M-1)!} \nonumber \\
&\left(\frac{\Gamma(p^{*}-q^{*})\Gamma(r^{*}-q^{*})}
{\Gamma\!\left(\omega-q^{*}+\frac12\right)}
\right)^M 
\left(\ln \left(\Phi^M\right)\right)^{M-1} \Phi^{Mq^{*}}.
\end{align}
By substituting this result into \eqref{up_bound_ir}, while using the fact that $\Phi\to 0^{+}$ and hence $(-1)^{M-1}(\ln\Phi^M)^{M-1}=M^{M-1}(-\ln\Phi)^{M-1} >0$, the high-SNR asymptotic behavior of the HARQ-IR upper bound is finally obtained as \eqref{high-snr-ir}.

\vspace{-0.1cm}
\section*{Appendix F}
\vspace{-0.1cm}
\label{int_1}
In order to validate \eqref{eq:pdf} as a PDF, we verify that the area under the entire curve is  equal to 1. 
Accordingly by using \cite[07.34.21.0009.0]{wolf}, we get 
\begin{align}
   &\int_0^{\infty} f_{h}(h) \mathrm{d}h = \frac{\sqrt{\omega}\alpha\beta}{h_l A_0 \Gamma(\alpha) \Gamma(\beta)} \nonumber\\
      & \times \int_0^{\infty}\mathrm{G}_{1,3}^{3,0}\!\left(\!\alpha \beta \frac{h}{A_0 h_l} \! \mathrel{\Bigg|}  
    \begin{array}{c}
	   \!\!\! \omega - \frac{1}{2} \\ \!\!\! \alpha - 1,\, \beta - 1,\, \omega-1
	\end{array}\!\!\!\right) \mathrm{d}h \nonumber \\
    & = \frac{\sqrt{\omega}\alpha\beta}{h_l A_0 \Gamma(\alpha) \Gamma(\beta)}\frac{A_0 h_l\Gamma(\omega)\Gamma(\alpha) \Gamma(\beta)}{ \alpha 
    \beta \Gamma(\omega+ \frac{1}{2})} = 1,
\end{align}
where we exploit the following Stirling approximation $\frac{\Gamma(\omega)}{\Gamma(\omega+\frac{1}{2})} \approx \frac{1}{\sqrt{\omega}}$ for $\omega \gg 1$, as presented also in \eqref{stirling_approx}. Along with the fact that $f_h(h)\geq 0$, we arrive at the result.








\bibliographystyle{IEEEtran}
\bibliography{IEEEabrv,References}

\begin{thebibliography}{10}
\providecommand{\url}[1]{#1}
\csname url@samestyle\endcsname
\providecommand{\newblock}{\relax}
\providecommand{\bibinfo}[2]{#2}
\providecommand{\BIBentrySTDinterwordspacing}{\spaceskip=0pt\relax}
\providecommand{\BIBentryALTinterwordstretchfactor}{4}
\providecommand{\BIBentryALTinterwordspacing}{\spaceskip=\fontdimen2\font plus
\BIBentryALTinterwordstretchfactor\fontdimen3\font minus
  \fontdimen4\font\relax}
\providecommand{\BIBforeignlanguage}[2]{{%
\expandafter\ifx\csname l@#1\endcsname\relax
\typeout{** WARNING: IEEEtran.bst: No hyphenation pattern has been}%
\typeout{** loaded for the language `#1'. Using the pattern for}%
\typeout{** the default language instead.}%
\else
\language=\csname l@#1\endcsname
\fi
#2}}
\providecommand{\BIBdecl}{\relax}
\BIBdecl

\bibitem{FSO_theory}
M.~A. Khalighi and M.~Uysal, ``Survey on free space optical communication: A
  communication theory perspective,'' \emph{IEEE Commun. Surveys Tuts.},
  vol.~16, no.~4, pp. 2231--2258, 4th Quart. 2014.

\bibitem{FSO_theory_2}
S.~A. Al-Gailani, M.~F. Mohd~Salleh, A.~A. Salem, R.~Q. Shaddad, U.~U. Sheikh,
  N.~A. Algeelani, and T.~A. Almohamad, ``A survey of free space optics ({FSO})
  communication systems, links, and networks,'' \emph{IEEE Access}, vol.~9, pp.
  7353--7373, Dec. 2020.

\bibitem{Haas_Survey}
A.~Krishnamoorthy~\textit{et al.}, ``Optical wireless communications: Enabling
  the next-generation network of networks,'' \emph{IEEE Veh. Technol. Mag.},
  vol.~20, no.~2, pp. 20--39, Jun. 2025.

\bibitem{FSO_app}
W.~Jiang, B.~Han, M.~A. Habibi, and H.~D. Schotten, ``The road towards {6G}:
  {A} comprehensive survey,'' \emph{IEEE Open J. Commun. Soc.}, vol.~2, pp.
  334--366, Feb. 2021.

\bibitem{nature_1}
Q.~Wang, E.~Rogers, B.~Gholipour, C.-M. Wang, G.~Yuan, J.~Teng, and
  N.~Zheludev, ``Optically reconfigurable metasurfaces and photonic devices
  based on phase change materials,'' \emph{Nature Photonics}, vol.~10, 12 2015.

\bibitem{nature_2}
Y.~Zhao, Z.~Liu, C.~Li, W.~Jiao, S.~Jiang, X.~Li, J.~Duan, and J.~Li,
  ``Mechanically reconfigurable metasurfaces: fabrications and applications,''
  \emph{npj Nanophotonics}, vol.~1, 06 2024.

\bibitem{schober_survey}
V.~Jamali, H.~Ajam, M.~Najafi, B.~Schmauss, R.~Schober, and H.~V. Poor,
  ``Intelligent reflecting surface assisted free-space optical
  communications,'' \emph{IEEE Commun. Mag.}, vol.~59, no.~10, pp. 57--63, Oct.
  2021.

\bibitem{Magaz}
H.~Wang, Z.~Zhang, B.~Zhu, J.~Dang, and L.~Wu, ``Optical reconfigurable
  intelligent surfaces aided optical wireless communications: Opportunities,
  challenges, and trends,'' \emph{IEEE Wireless Commun.}, vol.~30, no.~5, pp.
  28--35, Oct. 2023.

\bibitem{roadmap}
A.~I. Kuznetsov~et al., ``Roadmap for optical metasurfaces,'' \emph{ACS
  Photonics}, vol.~11, no.~3, pp. 816--865, 2024, pMID: 38550347.

\bibitem{Optical_adaptive}
C.~Hail, A.-K. Michel, D.~Poulikakos, and H.~Eghlidi, ``Optical metasurfaces:
  Evolving from passive to adaptive,'' \emph{Advanced Optical Materials},
  vol.~7, 05 2019.

\bibitem{elec_tun}
\BIBentryALTinterwordspacing
F.~Ding, C.~Meng, and S.~I. Bozhevolnyi, ``{Electrically tunable optical
  metasurfaces},'' \emph{Photonics Insights}, vol.~3, no.~3, p. R07, 2024.
  [Online]. Available: \url{https://doi.org/10.3788/PI.2024.R07}
\BIBentrySTDinterwordspacing

\bibitem{beam_splitter}
\BIBentryALTinterwordspacing
T.~Tian, Y.~Liao, X.~Feng, K.~Cui, F.~Liu, W.~Zhang, and Y.~Huang,
  ``Metasurface-based free-space multi-port beam splitter with arbitrary power
  ratio,'' \emph{Advanced Optical Materials}, vol.~11, no.~20, p. 2300664,
  2023. [Online]. Available:
  \url{https://advanced.onlinelibrary.wiley.com/doi/abs/10.1002/adom.202300664}
\BIBentrySTDinterwordspacing

\bibitem{Alou_ORIS}
\BIBentryALTinterwordspacing
L.~Yang, W.~Guo, D.~B. da~Costa, and M.~Alouini, ``Free-space optical
  communication with reconfigurable intelligent surfaces,'' 2020. [Online].
  Available: \url{https://arxiv.org/abs/2012.00547}
\BIBentrySTDinterwordspacing

\bibitem{Unified}
V.~K. Chapala and S.~M. Zafaruddin, ``Unified performance analysis of
  reconfigurable intelligent surface empowered free-space optical
  communications,'' \emph{IEEE Trans. Commun.}, vol.~70, no.~4, pp. 2575--2592,
  2022.

\bibitem{Dobre_Haas}
A.~R. Ndjiongue, T.~M.~N. Ngatched, O.~A. Dobre, A.~G. Armada, and H.~Haas,
  ``Analysis of {RIS}-based terrestrial-{FSO} link over {G-G} turbulence with
  distance and jitter ratios,'' \emph{J. Lightw. Tech.}, vol.~39, no.~21, pp.
  6746--6758, Nov. 2021.

\bibitem{photonics_1}
T.~Ishida, C.~B. Naila, H.~Okada, and M.~Katayama, ``Performance analysis of
  {IRS}-assisted multi-link {FSO} system under pointing errors,'' \emph{IEEE
  Photonics J.}, vol.~16, no.~4, pp. 1--10, Aug. 2024.

\bibitem{boulo_fso}
A.-A.~A. Boulogeorgos, N.~D. Chatzidiamantis, H.~G. Sandalidis, A.~Alexiou, and
  M.~D. Renzo, ``Cascaded composite turbulence and misalignment: Statistical
  characterization and applications to reconfigurable intelligent
  surface-empowered wireless systems,'' \emph{IEEE Trans. Veh. Technol.},
  vol.~71, no.~4, pp. 3821--3836, Apr. 2022.

\bibitem{Najafi_new}
M.~Najafi, B.~Schmauss, and R.~Schober, ``Intelligent reflecting surfaces for
  free space optical communication systems,'' \emph{IEEE Trans. Commun.},
  vol.~69, no.~9, pp. 6134--6151, Sep. 2021.

\bibitem{optics_express}
B.~Chen and X.~Zhu, ``Outage-guaranteed transmission for {IRS}-assisted {FSO}
  systems,'' \emph{Opt. Express}, vol.~32, no.~14, pp. 25\,420--25\,434, Jul.
  2024.

\bibitem{ORIS-NOMA}
G.~D. Chondrogiannis, A.~P. Chrysologou, A.-A.~A. Boulogeorgos, N.~D.
  Chatzidiamantis, and H.~Haas, ``Optical ris-enabled multiple access
  communications,'' \emph{IEEE Transactions on Green Communications and
  Networking}, vol.~10, pp. 1068--1080, 2026.

\bibitem{HARQ_FSO}
E.~Zedini, A.~Chelli, and M.-S. Alouini, ``On the performance analysis of
  hybrid {ARQ} with incremental redundancy and with code combining over
  free-space optical channels with pointing errors,'' \emph{IEEE Photonics
  Journal}, vol.~6, no.~4, pp. 1--18, 2014.

\bibitem{improv}
A.~Touati, M.~O. Hasna, and F.~Touati, ``{HARQ} performance over {FSO} channels
  with atmospheric fading and pointing errors,'' in \emph{2018 14th
  International Wireless Communications \& Mobile Computing Conference
  (IWCMC)}, 2018, pp. 158--163.

\bibitem{Power-optimal_HARQ}
G.~D. Chondrogiannis, N.~A. Mitsiou, N.~D. Chatzidiamantis, A.-A.~A.
  Boulogeorgos, and G.~K. Karagiannidis, ``Power-optimal {HARQ} protocol for
  reliable free space optical communication,'' in \emph{2023 IEEE International
  Conference on Communications Workshops (ICC Workshops)}, 2023, pp.
  1765--1770.

\bibitem{harq-fso-vehicular}
C.~T. Nguyen, H.~D. Le, C.~T. Nguyen, and A.~T. Pham, ``Toward practical
  {HARQ}-based {RC-LDPC} design for optical satellite-assisted vehicular
  networks,'' \emph{IEEE Transactions on Aerospace and Electronic Systems},
  vol.~60, no.~6, pp. 8619--8634, 2024.

\bibitem{HARQ_FSO_UAV}
K.~D. Dang, H.~D. Le, C.~T. Nguyen, and A.~T. Pham, ``Cooperative {HARQ}-aided
  multiple {UAV}s in optical aerospace backhaul networks,'' \emph{IEEE Access},
  vol.~11, pp. 138\,247--138\,260, 2023.

\bibitem{Gradshteyn2014}
I.~S. Gradshteyn and I.~M. Ryzhik, \emph{Table of integrals, series, and
  products}.\hskip 1em plus 0.5em minus 0.4em\relax Academic press, 2014.

\bibitem{H-function_Mathai}
A.~Mathai, R.~Saxena, and H.~Haubold, \emph{The H-function: Theory and
  Applications}, 01 2009.

\bibitem{Hralinovic}
A.~A. Farid and S.~Hranilovic, ``Outage capacity optimization for free-space
  optical links with pointing errors,'' \emph{J. Lightw. Technol.}, vol.~25,
  no.~7, pp. 1702--1710, Jul. 2007.

\bibitem{Kim}
M.~C.~A. Naboulsi, H.~Sizun, and F.~de~Fornel, ``Fog attenuation prediction for
  optical and infrared waves,'' \emph{Opt. Eng.}, vol.~43, no.~2, pp. 319 --
  329, Feb. 2004.

\bibitem{Najafi_conf}
M.~Najafi and R.~Schober, ``Intelligent reflecting surfaces for free space
  optical communications,'' in \emph{Proc. IEEE Global Commun. Conf.
  (GLOBECOM)}, Waikoloa, HI, USA, Dec. 2019, pp. 1--7.

\bibitem{algorithm}
V.~S. Adamchik and O.~I. Marichev, ``The algorithm for calculating integrals of
  hypergeometric type functions and its realization in reduce system,'' in
  \emph{Proc. Int. Symp. Symbolic Algebr. Comput. (ISSAC)}, New York, NY, USA,
  1990, p. 212–224.

\bibitem{Diamantoulakis}
M.~Najafi, V.~Jamali, P.~D. Diamantoulakis, G.~K. Karagiannidis, and
  R.~Schober, ``Non-orthogonal multiple access for {FSO} backhauling,'' in
  \emph{Proc. IEEE Wireless Commun. Netw. Conf. (WCNC)}, Barcelona, Spain, Apr.
  2018, pp. 1--6.

\bibitem{E_D}
A.~Chelli, E.~Zedini, M.-S. Alouini, M.~Pätzold, and I.~Balasingham,
  ``Throughput and delay analysis of {HARQ} with code combining over double
  rayleigh fading channels,'' \emph{IEEE Transactions on Vehicular Technology},
  vol.~67, no.~5, pp. 4233--4247, 2018.

\bibitem{Integral_of_Fox}
W.~Ameen~Alathwary and E.~S. Altubaishi, ``An integral of {F}ox’s
  {H}-{F}unctions with application to the performance of hybrid {FSO/RF}
  systems over generalized fading channels,'' \emph{IEEE Open J. Commun. Soc.},
  vol.~6, pp. 1030--1041, 2025.

\bibitem{KS}
F.~J. Massey, ``The kolmogorov-smirnov test for goodness of fit,''
  \emph{Journal of the American Statistical Association}, vol.~46, no. 253, pp.
  68--78, 1951.

\bibitem{Luke}
Y.~Luke, \emph{The Special Functions and Their Approximations}, ser.
  Mathematics in Science and Engineering.\hskip 1em plus 0.5em minus
  0.4em\relax Academic Press, 1969.

\bibitem{H-transforms}
A.~Kilbas and M.~Saigo, \emph{{H}-transforms: {T}heory and applications}, 2004.

\bibitem{wolf}
\emph{“The Mathematical Functions Site.”}, 2022. [Online]. Available:
  https://functions.wolfram.com/.

\end{thebibliography}
\end{document}